\documentclass[12pt]{article}

\usepackage{amsmath,amsfonts,amssymb,amsthm,graphicx,cite}
\usepackage{amscd}
\usepackage{mathrsfs}
\usepackage{undertilde}
\usepackage[all,cmtip]{xy}

\numberwithin{equation}{section}

\begin{document}

\newtheorem{theorem}{Theorem}[section]
\newtheorem{corollary}[theorem]{Corollary}
\newtheorem{lemma}[theorem]{Lemma}
\newtheorem{proposition}[theorem]{Proposition}

\newcommand{\be}{\begin{equation}}
\newcommand{\ee}{\end{equation}}
\newcommand{\bea}{\begin{eqnarray}}
\newcommand{\eea}{\end{eqnarray}}
\newcommand{\De}{\Delta}
\newcommand{\de}{\delta}
\newcommand{\Z}{{\mathbb Z}}
\newcommand{\N}{{\mathbb N}}
\newcommand{\C}{{\mathbb C}}
\newcommand{\Cs}{{\mathbb C}^{*}}
\newcommand{\R}{{\mathbb R}}
\newcommand{\Q}{{\mathbb Q}}
\newcommand{\T}{{\mathbb T}}
\newcommand{\cW}{{\cal W}}
\newcommand{\cJ}{{\cal J}}
\newcommand{\cE}{{\cal E}}
\newcommand{\cA}{{\cal A}}
\newcommand{\cR}{{\cal R}}
\newcommand{\cP}{{\cal P}}
\newcommand{\cM}{{\cal M}}
\newcommand{\cN}{{\cal N}}
\newcommand{\cI}{{\cal I}}
\newcommand{\cB}{{\cal B}}
\newcommand{\cD}{{\cal D}}
\newcommand{\cC}{{\cal C}}
\newcommand{\cL}{{\cal L}}
\newcommand{\cF}{{\cal F}}
\newcommand{\cH}{{\cal H}}
\newcommand{\cS}{{\cal S}}
\newcommand{\cT}{{\cal T}}
\newcommand{\cU}{{\cal U}}
\newcommand{\cQ}{{\cal Q}}
\newcommand{\cV}{{\cal V}}
\newcommand{\cK}{{\cal K}}
\newcommand{\intR}{\int_{-\infty}^{\infty}}
\newcommand{\diag}{{\rm diag}}
\newcommand{\Ln}{{\rm Ln}}
\newcommand{\Arg}{{\rm Arg}}
\newcommand{\pz}{\partial_z}
\newcommand{\re}{{\rm Re}\, }
\newcommand{\im}{{\rm Im}\, }
\newcommand{\res}{{\rm Res}}

\title{Difference operators of Sklyanin and van Diejen type}
\author{Eric Rains \\Department of Mathematics, \\California Institute of Technology, Pasadena, USA \\ and \\Simon Ruijsenaars \\ School of Mathematics, \\ University of Leeds, Leeds LS2 9JT, UK}

\date{}

\maketitle

\begin{abstract}
  The Sklyanin algebra $\cS_{\eta}$ has a well-known family of
  infinite-dimensional representations $\cD(\mu)$, $\mu \in \Cs$, in terms
  of difference operators with shift $\eta$ acting on even meromorphic
  functions. We show that for generic $\eta$ the coefficients of these
  operators have solely simple poles, with linear residue relations
  depending on their locations. More generally, we obtain explicit
  necessary and sufficient conditions on a difference operator for it to
  belong to $\cD(\mu)$. By definition, the even part of $\cD(\mu)$ is
  generated by twofold products of the Sklyanin generators. We prove that
  any sum of the latter products yields a difference operator of van Diejen
  type. We also obtain kernel identities for the Sklyanin generators. They
  give rise to order-reversing involutive automorphisms of $\cD(\mu)$, and
  are shown to entail previously known kernel identities for the van Diejen
  operators. Moreover, for special $\mu$ they yield novel
  finite-dimensional representations of $\cS_{\eta}$.
\end{abstract}

\tableofcontents


\section{Introduction}
The Sklyanin algebra $\cS_{\eta}$~\cite{Skl82,Skl83} can be defined via four generators $S_0,S_1,S_2,S_3$ satisfying six relations
\be\label{Sk1}
[S_0,S_k]_-=iJ_{lm}(\eta)[S_l,S_m]_+,
\ee
\be\label{Sk2}
[S_k,S_l]_-=i[S_0,S_m]_+.
\ee
Here, $(k,l,m)$ is a cyclic permutation of $(1,2,3)$, and the structure constants $J_{23},J_{31}$ and $J_{12}$ are elliptic functions of $\eta$. Throughout, the elliptic lattice will be fixed as
\be\label{Lam}
\Lambda =\Z + \tau \Z,\ \ \ \im \tau>0.
\ee

The Sklyanin algebra has been studied from various perspectives and has been generalized in more than one way. There is meanwhile a considerable literature on this subject, from which we mention specifically Refs.~\cite{ATB90,CLOZ08,KZ95,OF93,Ode03,Ros04,SS93}.
This paper is concerned with representations of the Sklyanin algebra which involve analytic difference operators. These operator representations were introduced by Sklyanin in~\cite{Skl83}, restricting attention to special values of the representation parameter for which the operators leave certain finite-dimensional spaces invariant. In this paper we encounter additional special values yielding finite-dimensional modules, but our focus is on generic parameters and infinite-dimensional representations. 

One of our main goals is to characterize the difference operators that arise in the latter representations. The Sklyanin generator~$S_t$ is represented by a difference operator~$D_t$ characterized by a meromorphic coefficient~$f_t(z)$,
which has period~1 and quasi-period~$\tau$ with multiplier~$\mu$. The difference operators are defined on the space~$\cM_e$ of meromorphic even functions and are of the form
\be\label{Dt}
(D_tF)(z)=f_t(z)F(z+\eta)+f_t(-z)F(z-\eta),\ \ \ F\in\cM_e, \ \ \ \eta\in\C^*,\ \ \ t=0,1,2,3.
\ee
Thus they leave the space~$\cM_e$  invariant. (See~Section~2 for the details of their definition.)
As a rule, the parameters on which the coefficients depend (in particular the representation label~$\mu$) will be suppressed, unless confusion might arise. 

In the sequel, an operator action such as~\eqref{Dt} will be abbreviated as
\be
D_t=f_t(z)\exp(\eta \pz)+(z\to -z).
\ee
We denote the four-dimensional vector space $\cV_1(\mu)$ of analytic difference operators (henceforth A$\De$Os) spanned by $D_0,D_1,D_2$ and $D_3$ by $\cV_1(\mu)$, and the associated representation of~$\cS_{\eta}$ by~$\cD(\mu)$.  We shall mostly work with A$\De$Os of the form
\be\label{AR}
A_R=f_R(z)\exp(\eta \pz)+(z\to -z),
\ee
that yield $\cV_1(\mu)$ as the coefficients $f_R(z)$ vary over a four-dimensional vector space~$V_1(\mu)$, cf.~the paragraph containing~\eqref{fR}. These A$\De$Os were introduced by the first-named author in~\cite{Rai06}. 

Clearly, the linear combinations of $k$-fold products of the operators $A_R$ are of the form
\be\label{AVk}
A^{(k)}=\sum_{m=0}^k c^{(k)}_{k-2m}(\eta;z)\exp((k-2m)\eta\partial_z),\ \ k>1.
\ee
We denote the subspace consisting of the A$\De$Os~\eqref{AVk} 
by $\cV_k(\mu)$. It is plain that the operators in these subspaces of the representation $\cD(\mu)$ are not periodic in~$\eta$, whereas the Sklyanin algebra~$\cS_{\eta}$ is elliptic in~$\eta$. As such, the representation and its subspaces have an additional dependence on the choice of~$\eta$, which we shall not make explicit in our notation, just as the dependence on $\tau$ is suppressed. 

A principal result of this paper consists in necessary and sufficient conditions on an A$\De$O of the form
\be\label{Agen}
A=\sum_{j\in\Z}c_j(z)\exp(j\eta \pz),
\ee
with finitely many coefficients $c_j$ nonzero, to belong to $\cD(\mu)$,
assuming no multiple of $\eta$ belongs to the elliptic lattice:
\be\label{etareq}
\eta \notin \Q \Lambda.
\ee
 The coefficient conditions are of a quite explicit nature:  the $c_j(z)$ are meromorphic functions satisfying
 \be\label{cmul}
c_j(z+1)=c_j(z),\ \ \ c_j(z+\tau)=\mu^jc_j(z),
\ee
 and related by
 \be\label{cjeven}
 c_j(-z)=c_{-j}(z),
 \ee
 whereas their poles are constrained by requiring that they be simple and occur only at points of the form
 \be
 z=z_e/2-\ell\eta,\ \ \ z_e\in\Lambda,\ \ \ \ell\in\Z,
 \ee
  with linear relations among the residues depending on the pole locations, cf.~(ii)--(iii) in Lemma~3.2. We refer to such A$\De$Os as A$\De$Os of {\it Sklyanin type}.
  
As just defined, the notion of Sklyanin type A$\De$O refers to a fixed $\eta$-value that does not belong to the dense set $\Q\Lambda$. On the other hand, it is already clear from~\eqref{AR}--\eqref{AVk} that the coefficients in~\eqref{AVk} are meromorphic in their dependence on~$\eta$. (Indeed, the functions~$f_R(z)$ are $\eta$-independent by definition.) We also obtain various results that have a bearing on this $\eta$-dependence. In particular, we shall see that the coefficients in~\eqref{AVk} have at most simple $z$-poles when we only require
\be\label{etark}
\eta\notin \cup_{\ell =1}^k (2\ell)^{-1}\Lambda,
\ee
as opposed to the restriction~\eqref{etareq}. (The latter is necessary to ensure that any A$\De$O in $\cD(\mu)$ have coefficients with at most simple $z$-poles, cf.~Lemma~3.2.) More generally, we shall prove that multiplication of the coefficient $c^{(k)}_{k-2m}(\eta;z)$ in~\eqref{AVk} by a theta function product (depending on $k$ and $m$) yields a function that is holomorphic in~$\eta$ and~$z$. Accordingly, we obtain an explicit picture of the $z$-poles of the coefficients in their dependence on~$\eta$.

Clearly, any A$\De$O in $\cV_2(\mu)$ is of the form
\be
A_{D}=c_2(z)\exp(2\eta \pz)+(z\to -z)+c_0(z),
\ee
with the coefficient $c_0$ an even function. From its definition it will also be obvious that $c_0$ is actually an elliptic function. We shall show that all of these A$\De$Os are of the type introduced by van Diejen~\cite{vDie94}, and that a given van Diejen A$\De$O belongs to the space $\cV_2(\mu)$ for a unique $\mu\in\Cs$. (In point of fact, the A$\De$Os $A_R$ can also be viewed as van Diejen A$\De$Os of a special type, as explained at the end of~Subsection~4.3.)

The van Diejen A$\De$Os satisfy kernel identities of the form
\be\label{KD}
(A_{D}(z)-A_{D}'(y))K(z,y)=cK(z,y),\ \ \ c\in\C,
\ee
Here, the constant $c$ depends on the convention for the additive constants
in the A$\De$Os, and  the prime signifies that the parameters of the
$y$-dependent A$\De$O differ (in general) from those of the $z$-dependent
one. Also, the kernel function $K(z,y)$ is a product of elliptic gamma
functions~\cite{Rui09}, cf.~also~\cite{KNS09} and \cite{Rui05}. A second main result of this paper is that when the A$\De$O~$A_D(z)$ is viewed as belonging to $\cV_2(\mu)$, then the same kernel function also serves as such for the A$\De$Os~$A_R(z)$ in $\cV_1(\mu)$: 
\be\label{KR}
(A_{R}(z)-A_{R}'(y))K(z,y)=0.
\ee
(The operators $A_D'(y)$ and $A_R'(y)$ also belong to~$\cV_2(\mu)$ and~$\cV_1(\mu)$, resp.)
Moreover, the kernel identity~\eqref{KD} can be viewed as a consequence of~\eqref{KR}.

We proceed with a more detailed sketch of the results and organization of this paper. Section~2 has a preparatory character. We introduce notation used throughout the paper and define various operators and spaces involving theta functions. In particular, Proposition~2.1 encodes a key description of a space of meromorphic functions satisfying certain quasi-periodicity and holomorphy restrictions. Specifically, this space can be viewed as a $4k$-dimensional vector space $V_k(\mu)$ of theta function ratios.

In Section~3 we obtain various insights into the structure of the subspaces $\cV_k(\mu)$. For the $\eta$-values satisfying the restriction~\eqref{etareq}, we arrive at the explicit characterization of the Sklyanin type A$\De$Os defined above via several lemmas. These lemmas contain additional information for general $\eta$-values. The characterization is encoded in Theorem~3.7, whereas Theorem~3.9 collects results concerning the $\eta$-dependence of the coefficients~$c^{(k)}_{k-2m}(\eta;z)$ in~\eqref{AVk}, including a remarkable quasi-periodicity feature.

In Subsection~4.1 we first collect notation and some results associated with the elliptic gamma function~$G(r,a_+,a_-;z)$ introduced in~\cite{Rui97}. With its real period $\pi/r$ normalized to 1, it serves as the building block for the kernel function $K(z,y)$. It is symmetric under interchange of $a_+$ and $a_-$ (`modular invariant'), and since $\tau $ and $\eta$ correspond to $ia_+$ and $ia_-/2$, a second copy of the Sklyanin algebra naturally arises. Hence we arrive at two (non-commuting) Sklyanin algebras $\cS_{\pm}$. These  algebras and their amalgamation were arived at before and studied in some detail by Spiridonov~\cite{Spi09}.

In Subsection~4.2 we obtain the kernel identities~\eqref{KR}. More precisely, in Theorem~4.1 we reformulate the identities so that they apply to the generating A$\De$Os of both Sklyanin algebras $\cS_+$ and $\cS_-$ at once. They give rise to two distinct order-reversing automorphisms of the algebras. 

Choosing special values for the parameter $\mu$, the kernel function becomes a product of theta functions. From the kernel identities it can then be deduced that the Sklyanin algebras~$\cS_{\pm}$ leave an associated finite-dimensional vector space of theta functions invariant.  For the $\mu$-choice $\exp(-2N\pi a_{+})$ with $N$ a nonnegative integer and for $\cS_{+}$ these spaces amount to the finite-dimensional modules studied by Sklyanin~\cite{Skl83} (cf.~also~\cite{Ros04}), but for~$\cS_{-}$ these modules are of a different type. Indeed, reverting to the single algebra $\cS_{\eta}$,  the building block of Sklyanin's modules  is the theta function $\theta_1(z|\tau)$ with quasi-period $\tau$, whereas for the latter modules it is the theta function $\theta_1(z|2\eta)$ with quasi-period $2\eta$. For the more general $\mu$-choices
\be
\mu=\exp (-2\pi(Ma_{-}+Na_{+})),\ \ \ M,N=1,2,3,\ldots,
\ee
there still exist finite-dimensional submodules of~$\cM_e$ for $\cS_{+}$ and $\cS_{-}$. It would be of interest to study these further. In particular, it is not obvious how these modules fit in the classification of finite-dimensional modules given by Smith and Staniszkis~\cite{SS93}.

Subsection~4.3 deals with the van Diejen A$\De$Os and their relation to the Sklyanin algebras $\cS_{\pm}$. The (modular generalization of the) identities~\eqref{KD} are obtained in Theorem~4.4, and the relation to the pertinent results of~\cite{Rui09} is established. We also discuss a remarkable consequence of the reinterpretation of the A$\De$Os $A_R$ as van Diejen A$\De$Os.

In Appendix~A we present the proof of a key lemma (Lemma~3.2), while in Appendix~B we focus on the connection between the Sklyanin relations~\eqref{Sk1}--\eqref{Sk2} and the relations between the generators~$A_R$ obtained in~\cite{Rai06}. Their equivalence for $\eta\notin \Lambda/2$ is explicitly established here for the first time, cf.~Theorem~B.4. We also discuss the state of affairs for $2\eta\in\Lambda$ (an $\eta$-choice that is usually excluded in the literature), and add a few remarks on the representations $\cD(\mu)$. 


\section{Preliminaries}
The above difference operators can all be defined in terms of theta functions. Various notations and conventions for theta functions can be found in the literature, and we need to specify our choice. We shall work at first with the theta functions used in particular by Rosengren in~\cite{Ros04}, but switch in Section~4 to a building block that is more convenient when working with the elliptic gamma function and van Diejen type A$\De$Os~\cite{Rui04}. 

Recalling our convention~\eqref{Lam} for the elliptic lattice $\Lambda$, the four Jacobi theta functions can be defined starting from
the building block
\be\label{th1}
\theta (z)\equiv \theta_1(z|\tau)=i\sum_{n\in\Z}(-)^nq^{(n-1/2)^2}\exp(i\pi (2n-1)z),\ \ \ q\equiv \exp(i\pi\tau),
\ee
which is odd, entire, satisfies
\be\label{aut}
\theta(z+1)=-\theta(z),\ \ \ \theta(z+\tau)=-q^{-1} \exp(-2i\pi z)\theta(z),
\ee
and has its zeros in the elliptic lattice points.
Specifically, the remaining theta functions are given by
\be\label{th3}
\theta_3(z)\equiv q^{1/4}\exp(i\pi z)\theta(z+1/2+\tau/2),
\ee
\be\label{th24}
\theta_2(z)\equiv \theta(z+1/2),\ \ \ \theta_4(z)\equiv\theta_3(z+1/2).
\ee
The structure constants of the Sklyanin algebra can now be expressed as
\be\label{Jlm}
J_{23}=\left( \frac{\theta_1\theta_2}{\theta_3\theta_4}\right)^2(\eta),\ \ 
J_{31}=-\left( \frac{\theta_1\theta_3}{\theta_2\theta_4}\right)^2(\eta),\ \ 
J_{12}=\left( \frac{\theta_1\theta_4}{\theta_2\theta_3}\right)^2(\eta),
\ee
from which their ellipticity in $\eta$ is readily checked. 

As has become customary, we use notation exemplified by
\be
\theta(a\pm b)=\theta(a+b)\theta(a-b),
\ee
\be\label{prnot}
\theta(a_1,\ldots,a_n)=\prod_{m=1}^n\theta(a_m),
\ee
\be
\theta(z+\vec{a})=\prod_{m=1}^n\theta(z+a_m),\ \ \ z\in\C,\ \ a\in\C^n.
\ee
Using the product notation~\eqref{prnot}, the duplication formula for $\theta(z)$ reads
\be\label{dup}
\theta(2z)=iq^{1/4}G^{-3}\theta(z,z+1/2,z+\tau/2,z-1/2-\tau/2),
\ee
where we have set
\be
G\equiv \prod_{m=1}^{\infty}(1-q^{2m}).
\ee

The coefficient function $f_R(z)$ of the A$\De$O $A_R$~\eqref{AR} is a meromorphic function with at most simple poles for $z\in \Lambda/2$ and no poles for $z\notin \Lambda/2$, which satisfies
\be\label{fR}
f_R(z+1)=f_R(z),\ \ \ f_R(z+\tau)=\mu f_R(z),\ \ \ \mu\in\Cs.
\ee
The vector space spanned by functions with these properties will be denoted by $V_1(\mu)$. In particular, it follows by using~\eqref{dup} and~\eqref{aut} that any function of the form
\be\label{fRth}
f(a,\nu;z)\equiv\theta(z+\vec{a}-\nu)/\theta(2z),\ \ a\in\C^4, \ \ \nu\in\C,
\ee
where
\be\label{anu}
\sum_{i=1}^4a_i=0,\ \ \ \exp(8i\pi\nu)=\mu,\ \ \re \nu\in[0,1/4),
\ee
belongs to $V_1(\mu)$.
(In fact, all $f_R\in V_1(\mu)$ are multiples of a function of the form~\eqref{fRth}--\eqref{anu}, and $V_1(\mu)$  is four-dimensional, cf.~Proposition~2.1 below.) 

The A$\De$Os $D_t$ representing the Sklyanin generators~$S_t$ are related to the A$\De$Os
\be\label{Apar}
A(a,\nu)\equiv f(a,\nu;z)\exp(\eta \pz)+(z\to -z),
\ee
 as follows:
\be\label{D0}
D_0=iq^{1/4}G^{-3}\theta(\eta)A((0,1,\tau,-1-\tau)/2,\nu),
\ee
\be
D_1=-iq^{1/4}G^{-3}\theta(\eta+1/2)A((1,-1,1+2\tau,-1-2\tau)/4,\nu),
\ee
\be
D_2=iq^{1/4}G^{-3}\exp(i\pi\eta)\theta(\eta+1/2+\tau/2)A((1+\tau,1-\tau,-1+\tau,-1-\tau)/4,\nu),
\ee
\be\label{D3}
D_3=iq^{1/4}G^{-3}\exp(i\pi\eta)\theta(\eta+\tau/2)A((\tau,-\tau,2+\tau,-2-\tau)/4,\nu).
\ee
Conversely, in the appendix of Rosengren's paper~\cite{Ros04} an explicit formula can be found for  $A(a,\nu)$ as a linear combination of the $D_t$, which we have no occasion to use here. We come back to the resulting representations $\cD(\mu)$ of $\cS_\eta$ in Appendix~B, where we also clarify the connection between the Sklyanin relations and the relations between the A$\De$Os $A(a,\nu)$ that were obtained by the first-named author in~\cite{Rai06}.

In the next section we shall see that  a consideration of $k$-fold products of the A$\De$Os $A(a,\nu)$ leads to a vector space $V_k(\mu)$ consisting of meromorphic functions $g(z)$ satisfying
\be
g(z+1)=g(z),\ \ \ \ g(z+\tau)=\mu^kg(z),
\ee
and such that the product
\be\label{holreq}
g(z)P_k(\eta,z),
\ee
with
\be\label{defPk}
P_k(\eta,z)\equiv \prod_{\ell=0}^{k-1}\theta (2z+2\ell\eta),\ \ \ \eta\in\C,
\ee
is holomorphic. The following proposition yields a more explicit picture of this space. 

\begin{proposition}
The vector space $V_k(\mu)$ is $4k$-dimensional and any $g\in V_k(\mu)$ can be written as
\be\label{Vkform}
g(z)=c\theta(z+\vec{a}-\nu)/P_k(\eta,z),\ \ \ c\in\C,\ \ \ a\in\C^{4k},
\ee
with
\be\label{asum}
\sum_{i=1}^{4k}a_i=2k(k-1)\eta,\ \ \ \exp(8i\pi\nu)=\mu,\ \ \re \nu\in[0,1/4).
\ee
\end{proposition}
\begin{proof}
We fix $4k$ numbers $a^{(0)}_1,\ldots,a^{(0)}_{4k}$ that are pairwise incongruent and satisfy
\be
\sum_{i=1}^{4k}a^{(0)}_i=2k(k-1)\eta.
\ee
Then it is easily verified that
\be
f_0(z)\equiv \theta(z+\vec{a}^{(0)}-\nu)/P_k(\eta,z)
\ee
belongs to $V_k(\mu)$. Now let $g(z)\in V_k(\mu)$ and consider the ratio $g(z)/f_0(z)$. This is an elliptic function with at most simple poles at the $4k$ pairwise incongruent numbers $\nu-a^{(0)}_1,\ldots,\nu-a^{(0)}_{4k}$. The space of elliptic functions with this property is $4k$-dimensional, since $4k-1$ residues and a constant can be freely chosen. Moreover, any function in this space can be factorized as
\be\label{ratio}
c\prod_{j=1}^{4k}\frac{\theta(z+z_j)}{\theta(z+a_j^{(0)}-\nu)},\ \ \ c\in\C,\ \ \ \sum_{j=1}^{4k} z_j\equiv 2k(k-1)\eta -4k\nu \pmod 1.
\ee
Setting $\tilde{a}_j:= \nu+z_j$, we need only shift one of the components of $\tilde{a}$ by a suitable integer to obtain a vector $a$ satisfying~\eqref{asum}. Then~\eqref{ratio} becomes
\be
c\prod_{j=1}^{4k}\frac{\theta(z+a_j-\nu)}{\theta(z+a_j^{(0)}-\nu)},
\ee
and the assertions easily follow.
\end{proof}


\section{A$\De$Os of Sklyanin type}
In this section we aim to characterize the A$\De$Os that belong to the representation $\cD(\mu)$ of the Sklyanin algebra~$\cS_{\eta}$. We start from an A$\De$O of the general form~\eqref{Agen} and obtain first necessary conditions for it to belong to $\cD(\mu)$ with the $\eta$-constraint~\eqref{etareq} in effect. To this end we begin by deriving features shared by all A$\De$Os~$A(a,\nu)$ (given by~\eqref{fRth}--\eqref{Apar}). First we introduce
\be
\omega_0=0,\ \ \omega_1=1/2,\ \ \omega_2=1/2+\tau/2,\ \ \omega_3=\tau/2,
\ee
\be
\lambda_0=\lambda_1= 0,\ \ 
\lambda_2=\lambda_3= 8i\pi\nu.
\ee
The properties~\eqref{fR} of the functions $f(a,\nu;z)$ given by~\eqref{fRth}--\eqref{anu} can be rewritten as
\be\label{faut}
f(a,\nu;z+\omega_t)=\exp(\lambda_t)f(a,\nu;z-\omega_t),\ \ \ t=0,1,2,3.
\ee
From this we readily deduce that the residues of the functions $f(a,\nu;\pm z)$ at the simple pole $z=\omega_t$ are related by
\be\label{fres}
\res_{z=\omega_t}f(a,\nu;z)=-\exp(\lambda_t)\res_{z=\omega_t}f(a,\nu;-z).
\ee

Next, we define four vector spaces~$\cM_t$, consisting of meromorphic functions $g(z)$ that are regular at all points $\omega_t+k\eta$, $k\in\Z$, and that satisfy
\be\label{gaut}
g(\omega_t-z)=\exp(\lambda_tz/\eta)g(\omega_t+z),\ \ \ t=0,1,2,3.
\ee
In particular, $\cM_0$ consists of the even meromorphic functions that have no poles at any integer multiple of $\eta$. The other three spaces are `equally large', in the sense that the maps
\be\label{Phit}
\Phi_t\, :\, \cM_t\to \cM_0,\ \ \ g(z)\mapsto h(z)=\exp (\lambda_t z/2\eta)g(z+\omega_t),\ \ \ t=1,2,3,
\ee
are easily seen to be bijections. Our use of the spaces $\cM_t$ is tied to the restriction~\eqref{etareq} on~$\eta$, which encompasses the $k$-dependent restrictions~\eqref{etark}.
We are now prepared for the following lemma.

\begin{lemma}
Assuming~\eqref{etareq}, the A$\De$Os $A(a,\nu)$ leave the above spaces invariant:
\be
A(a,\nu)\cM_t\subset \cM_t,\ \ t=0,1,2,3.
\ee
\end{lemma}
\begin{proof}
Letting $g\in\cM_t$, consider the function
\be
(A(a,\nu)g)(z)=f(a,\nu;z)g(z+\eta)+f(a,\nu;-z)g(z-\eta).
\ee
Since $g\in\cM_t$, the functions $g(z\pm \eta)$ have no poles for $z-\omega_t\in \eta\Z$. In view of the $\eta$-constraint~\eqref{etareq}, the functions $f(a,\nu;\pm z)$ have no poles for $z-\omega_t\in \eta\Z^*$, whereas they have at most a simple pole at $z=\omega_t$. However, $(A(a,\nu)g)(z)$ has no pole at $z=\omega_t$, since the residues of the two terms on the right-hand side cancel due to~\eqref{fres} and~\eqref{gaut}.

It remains to show that the function~$A(a,\nu)g$ has the automorphy property~\eqref{gaut}. This is easily verified by combining this property for $g$ with~\eqref{faut}.
\end{proof}

Since any A$\De$O in $\cV_1(\mu)$ is of the form $cA(a,\nu)$ with $c\in\C$ and $a,\nu$ satisfying~\eqref{anu}, it follows from this lemma that with~\eqref{etareq} in force the four spaces $\cM_t$ are left invariant by all of the A$\De$Os belonging to the representation~$\cD(\mu)$ of $\cS_{\eta}$. In the next lemma we characterize A$\De$Os of the form~\eqref{Agen} that have this property. Its proof is somewhat long and technical, so we have relegated it to Appendix~A. 
\begin{lemma}
Assume $\eta$ satisfies \eqref{etareq}. Let
\be
A=\sum_{j=-k}^{k}c_j(z)\exp(j\eta \pz),
\ee
where $k$ is a positive integer and the coefficients are meromorphic functions. Then we have
\be\label{AMt}
A\cM_t\subset \cM_t,\ \ t=0,1,2,3,
\ee
if and only if the coefficients~$c_j(z)$, $|j|\le k$, have the following three properties: 

\noindent
(i) They satisfy the symmetries
\be\label{i}
c_j(z+\omega_t)=\exp(j\lambda_t)c_{-j}(-z+\omega_t),\ \ \ t=0,1,2,3;
\ee

\noindent
(ii) At any point of the form 
\be\label{ii}
z(t,\ell):=\omega_t-\ell\eta,\ \ \  \ell\in\Z,
\ee
they have at most simple poles with residues $r_j(t,\ell)$;

\noindent
(iii) These residues satisfy
\be\label{iiia}
r_j(t,j)=0,\ \ \ r_j(t,\ell)= 0,\ \ \ |2\ell-j|>k,
\ee
\be\label{iiib}
r_j(t,\ell)=-\exp((j-\ell)\lambda_t )r_{2\ell -j}(t,\ell) ,\ \ \ |2\ell-j|\le k.
\ee
\end{lemma}

 Note that it follows from~\eqref{faut}--\eqref{fres} that the coefficient properties in the lemma are satisfied for the special case $A=A(a,\nu)$. Note also that when we set $\ell =j$ in~\eqref{iiib}, then it follows that $r_j(t,j)=0$. Likewise, the second vanishing property in~\eqref{iiia} can be viewed as a consequence of~\eqref{iiib}, provided we omit the restriction on $\ell$ and put $c_m(z)\equiv 0$ for $|m|>k$.
 
 Next, we obtain information on the space $\cV_2(\mu)$ of A$\De$Os that are linear combinations of twofold products of the A$\De$Os in $\cV_1(\mu)$. Letting
 \be\label{Af}
A(f)=f(z)\exp(\eta\partial_z)+(z\to -z),\ \ \ f\in V_1(\mu),
\ee
the product $A(f_1)A(f_2)$ is of the form 
\be\label{2form}
c_2(z)\exp(2\eta\partial_z)+c_0(z)+c_{-2}(z)\exp(-2\eta\partial_z),
\ee
with
\be\label{c0}
c_{\pm 2}(z)= f_1(\pm z)f_2(\pm z+\eta),\ \  c_0(z)= f_1(z)f_2(-z-\eta)+(z\to -z).
\ee
From this it is plain that the coefficients are meromorphic 1-periodic functions satisfying
\be
c_{\pm 2}(z+\tau)=\mu^{\pm 2}c_{\pm 2}(z),\ \ \ c_0(z+\tau)=c_0(z).
\ee
Moreover, assuming $\eta\notin\Lambda/2$, the coefficient $c_2(z)$ has at most simple poles for $z\in \Lambda/2$ and $z\in \Lambda/2-\eta$, and the coefficient properties (i)--(iii) in the previous lemma are easily verified directly.

More generally, it is plain that for all $\eta\in\C^{*}$ the general A$\De$O in~$\cV_2(\mu)$ is of the form~\eqref{2form}, with $c_2(z)$ in the space~$V_2(\mu)$ (defined above Prop.~2.1) and~$c_0(z)$ an elliptic function. 
On the other hand, it is not obvious, but true that any coefficient $c_2\in V_2(\mu)$ arises by taking suitable linear combinations of twofold products. For the special case $c_2(z)=0$ and $\eta$ satisfying~\eqref{etareq}, it follows from the previous lemma that $c_0(z)$ has no poles and hence is constant. This constant need not vanish, however. Indeed, constants arise for any $\eta\in\C^*$. In the following lemma we prove these two assertions.

\begin{lemma}
Let $\eta\in\C^*$. Then the constants form a subspace of~$\cV_2(\mu)$. Moreover, for any~$c_2\in V_2(\mu)$ there exists an A$\De$O in~$\cV_2(\mu)$ of the form~\eqref{2form}.
\end{lemma}
\begin{proof}
In order to show $\C\subset\cV_2(\mu)$, we define four functions
\be
e_j(z)=\theta(z+b_j)\theta(z-\alpha -b_j),\ \ e_{j+2}(z)=\theta(z+b_{j+2})\theta(z-\beta -b_{j+2}),\ \ j=1,2,
\ee
where the constants $b_1,\ldots,b_4$ are arbitrary and
\be\label{abres}
\alpha + \beta=4\nu+2\eta.
\ee
This entails that the four functions
\be\label{fag}
f(z)=e_1(z)e_3(z+\eta)/\theta(2z),\ \ \ g(z)=e_4(z)e_2(z+\eta)/\theta(2z),
\ee
and 
\be\label{fpagp}
f'(z)=e_1(z)e_4(z+\eta)/\theta(2z),\ \ \ g'(z)=e_3(z)e_2(z+\eta)/\theta(2z),
\ee
belong to $V_1(\mu)$, and that we have an equality
\be\label{fg2}
f(z)g(z+\eta)=f'(z)g'(z+\eta).
\ee
Hence we have
\be
A(f)A(g)-A(f')A(g')=c_0(z)-c_0'(z),
\ee
where (cf.~\eqref{2form}--\eqref{c0})
\be
c_0(z)= f(z)g(-z-\eta)+(z\to -z),\ \ \ c_0'(z)= f'(z)g'(-z-\eta)+(z\to -z).
\ee

Next, we introduce two functions
\be
F_j(z)=e_j(z)e_{j+1}(-z)/\theta(2z), \ \ \ j=1,3,
\ee
which are elliptic with at most simple poles for $z\in\Lambda/2$. Thus we have
\be
F_j(z)+F_j(-z)= k_j,\ \ j=1,3,
\ee
with some constants $k_1,k_3$. (Indeed, the residues at the four points $z=\omega_t$ vanish.) 
Now a straightforward calculation yields
\be
c_0(z)-c_0'(z)  =   -F_1(z)[F_3(z+\eta)+F_3(-z-\eta)] -F_1(-z)[F_3(-z+\eta)+F_3(z-\eta)].
\ee
Hence we have
\be
A(f)A(g)-A(f')A(g')=-k_1k_3.
\ee

We proceed to show that parameter choices exist such that the constants $k_1$ and $k_3$ do not vanish. First, we note that they can be written
\be
 k_1 = F_1(b_1)+F_1(-b_1) = F_1(b_1)
   = \theta(-\alpha)\theta(b_2-b_1)\theta(-\alpha-b_1-b_2),
\ee
\be
k_3 = F_3(b_3)+F_3(-b_3) = F_3(b_3)
   = \theta(-\beta)\theta(b_4-b_3)\theta(-\beta-b_3-b_4).
\ee
Since the constants $b_1-b_2$ and $b_3-b_4$ are at our disposal, we can choose them such that $\theta(b_2-b_1)$ and $\theta(b_4-b_3)$ do not vanish. Next, we can choose $\alpha$ and $\beta$ such that the remaining theta-factors do not vanish either, since $\alpha$ and $\beta$ are only constrained by~\eqref{abres}. Thus we have shown that $\cV_2(\mu)$ contains the constants. 

We now prove that for any~$c_2\in V_2(\mu)$ there exists an A$\De$O in~$\cV_2(\mu)$ of the form~\eqref{2form}. Clearly, the assertion amounts to the claim that the space $V_2(\mu)$ is spanned by products $f(z)g(z+\eta)$ with $f,g\in V_1(\mu)$. Our proof proceeds in two steps, revealing a remarkable dichotomy. First, we handle the case
\be\label{eL4}
\eta\notin\Lambda/4,
\ee
and then we consider nonzero eta's in $\Lambda/4$.  To prove the claim, we fix~$f_1,f_2\in V_1(\mu)$ with four simple zeros in a period cell and no common zeros. Then we consider the map
\be\label{kmap}
(g_1,g_2)\mapsto f_1(z)g_1(z+\eta)-f_2(z)g_2(z+\eta),
\ee 
where $g_1,g_2\in V_1(\mu)$. This is a linear map from the 8-dimensional space $V_1(\mu)\oplus V_1(\mu)$ into the 8-dimensional space~$V_2(\mu)$. To show that it is onto, we need only prove it has trivial kernel. 

Thus, let us assume we have
\be\label{fg}
f_1(z)g_1(z+\eta)=f_2(z)g_2(z+\eta),
\ee
with $g_1,g_2\ne 0$. Now $f_1(z)$ is of the form $cf(a,\nu;z)$ (cf.~\eqref{fRth}--\eqref{anu}), and its four zeros at $z=\nu -a_j$ are not among those of $f_2(z)$. Thus $g_2(z)$ must have zeros for $z=\nu -a_j+\eta$. On the other hand, $g_2(z)$ is of the form $c'f(a',\nu)$, so reshuffling the components of $a$ if need be, we must have four congruences
\be
\nu -a_j+\eta \equiv \nu -a'_j,\ \ j=1,2,3,4.
\ee
Summing over $j$, this implies $4\eta\equiv 0$, i.e., $4\eta\in \Lambda$. This contradicts our assumption~\eqref{eL4}, so we must have $g_1=g_2=0$, hence a trivial kernel.

Turning to the case of $\eta\in\Lambda/4$ (including 0), we can still show
that any function in $V_2(\mu)$ can be written as a linear combination of two products
$f(z)g(z+\eta)$ with $f,g\in V_1(\mu)$, which suffices to complete the
proof of the lemma.  
Unlike in the case $\eta\notin \Lambda/4$, however, a map of the kind just considered has a 1-dimensional
kernel, and thus the range of the map is a 7-dimensional subspace of
$V_2(\mu)$.

The 1-dimensional kernel arises from the choice of functions $g_1(z)$ and $g_2(z)$ in $V_1(\mu)$ whose zeros are given by those of $f_2(z)$ and $f_1(z)$, shifted by $\eta$. We give an example for $\eta=\tau/4$ to illustrate this.
Letting
\be\label{eg1}
f_j(z)=\theta(z+\vec{a}^{(j)}-\nu)/\theta(2z), \ \ \ \sum_{n=1}^4a_n^{(j)}=0, \ \ \ j=1,2,
\ee
the kernel consists of multiples of $(t_1(z)/\theta(2z),t_2(z)/\theta(2z))\in V_1(\mu)^2$, with $t_j$ given by
\be\label{eg2}
t_1(z)=\exp(-2\pi iz)\theta(z-\tau/4+\vec{a}^{(2)}-\nu),\ \ t_2(z)=\exp(-2\pi iz)\theta(z-\tau/4+\vec{a}^{(1)}-\nu).
\ee

It is unclear how a function in the 7-dimensional image for 
$f_1$, $f_2$ with the above restrictions can be recognized. But we can relax the requirements on $f_1$ and $f_2$ in such a way that we still obtain a 7-dimensional image that can be explicitly described, and this can be exploited to complete the proof.

In order to detail this, we fix attention on the period cell spanned by the numbers~1 and~$\tau$. Then
 we can either allow $f_1$ and $f_2$ to have four simple zeros and a unique common zero~$z_0$ in the cell, or to have three pairwise distinct zeros and no pole at~$z=0$. In either case,
the map still has 1-dimensional kernel. (Indeed, the locations of three zeros of $g_1$ and
$g_2$ in the cell follow from the constraint that $(g_1,g_2)$ be in the kernel, and
quasi-periodicity then renders $(g_1,g_2)$ unique up to a constant, cf.~again the example~\eqref{eg1}--\eqref{eg2}, now with~$a_1^{(1)}=a_1^{(2)}$, say.)  Moreover, in the first case every function in the image vanishes at~$z_0$. 
The subspace~$V(z_0)\subset V_2(\mu)$ of functions with this property is 7-dimensional,  since we need only supplement~$V(z_0)$ with a function in~$V_2(\mu)$ that is nonzero at~$z_0$ to obtain all of~$V_2(\mu)$ via linear combinations.Therefore, the image equals~$V(z_0)$. 

Likewise, in the second case every function in the image is analytic at $z=0$ for $\eta$ not congruent to 0, and has at most a simple pole at $z=0$ for $\eta\equiv 0$. Once more, 
this entails that the subspace~$V(0)\subset V_2(\mu)$ of functions with this property is 7-dimensional.  (Just as in the previous case, one of the eight zeros of the factor~$\theta(z+\vec{a}-\nu)$ in~\eqref{Vkform} is prescribed, namely, $z=0$.) Hence the image equals~$V(0)$.

There is now a corresponding case distinction depending both on $\mu$ and the function in~$V_2(\mu)$ we consider. Letting $\mu^2\ne\exp(2\pi i n\tau)$ with $n\in\Z$, any $f\in V_2(\mu)$ has at least one zero in the cell, so~$f$ belongs to $V(z_0)$ for some $z_0$. Hence~$f$ can be written as a linear combination of two products.
Next, assume~$\mu^2$ equals $\exp(2n\pi i \tau)$ with $n\in\Z$. Then it follows as before that any~$f\in V_2(\mu)$ having zeros is a linear combination of two products. However, setting
\be
e_k(z)\equiv \exp(k\pi iz),\ \ \ k\in\Z,
\ee
 the function~$e_{2n}(z)$ belongs to~$V_2(\mu)$ and has no zeros. Also, any~$f\in V_2(\mu)$ without zeros is a multiple of~$e_{2n}(z)$. 

To dispose of this last case, it suffices to note~$e_{2n}(z)$ belongs to $ V(0)$, since this implies that the special function~$e_{2n}(z)$ can also be written as a linear combination of two products. In fact,
 for $\mu=\exp(n\pi i \tau)$ with $n$ even, we have~$e_n(z)\in V_1(\mu)$, and so $e_{2n}(z)$ can already be obtained as a single product $f(z)g(z+\eta)$ with $f,g\in V_1(\mu)$. On the other hand, for~$\mu$ equal to $\exp(n\pi i \tau)$ with $n$ odd or equal to~$-\exp(n\pi i \tau)$ with $n\in\Z$, all functions in $V_1(\mu)$ have at least one zero, so that we need two products.  
\end{proof}

In order to handle $k$-fold products of the generators, we need the following lemma.

\begin{lemma}
Assume $\eta\in\C$ and $k>2$. Choose $f_1,f_2\in V_1(\mu)$ with four simple zeros in a period cell and no common zeros. Then any function $f$ in the vector space $V_k(\mu)$ (defined above Prop.~2.1) can be written as
\be\label{flc}
f(z)=f_1(z)g_1(z+\eta)-f_2(z)g_2(z+\eta),\ \ \ g_1,g_2\in V_{k-1}(\mu).
\ee
\end{lemma}
\begin{proof}
We reconsider the map~\eqref{kmap} in the proof of Lemma~3.3, now with $g_1,g_2\in V_{k-1}(\mu)$. The kernel of this linear map from the $(8k-8)$-dimensional space $V_{k-1}(\mu)\oplus V_{k-1}(\mu)$ into the $4k$-dimensional space $V_k(\mu)$ consists of pairs $(g_1,g_2)$ satisfying~\eqref{fg}. Thus the zeros of $f_2(z)$ and $f_1(z)$ are also zeros of $g_1(z+\eta)$ and $g_2(z+\eta)$, resp., and so we have
\be
g_1(z+\eta)/f_2(z)=g_2(z+\eta)/f_1(z)=
\theta (z+\vec{a})\theta(2z)/\prod_{n=1}^{k-1}\theta (2z+2n\eta),
\ee
where $a$ is a vector in~$\C^{4k-8}$. From this it readily follows that the kernel is $(4k-8)$-dimensional. (This follows e.~g.~by arguing as in the proof of Prop.~2.1.) Thus, the map is onto $V_k(\mu)$.
\end{proof}
 
Consider now the space $\cV_k(\mu)$ of A$\De$Os spanned by the $k$-fold products of the $A(f)$ with $f\in V_1(\mu)$.

\begin{lemma} 
Let $k>0$ and let $\eta\in\C^*$. Then all operators in $\cV_k(\mu)$ are of the form~\eqref{AVk},
where the coefficients are meromorphic functions that satisfy~\eqref{cmul}--\eqref{cjeven}. Also, we have
\be\label{ckV}
c^{(k)}_k(z)\in V_k(\mu),
\ee 
\be\label{subs}
\cV_{k-2}(\mu)\subset \cV_k(\mu),\ \ \ \cV_{-1}(\mu)\equiv \{ 0\},\ \ \ \cV_0(\mu)\equiv \C.
\ee
Finally, assuming $\eta$ is restricted by~\eqref{etareq}, the coefficients have the properties (i)--(iii) specified in Lemma~3.2 and the functions
\be\label{poles}
c^{(k)}_{k-2m}(z)\prod_{\ell=-k+1}^{k-1}\theta(2(z-\ell\eta)),\ \ \ m=0,\ldots,k,
\ee
are holomorphic.
\end{lemma}
\begin{proof}
It is plain that any operator in~$\cV_k(\mu)$ has the form~\eqref{AVk} with  meromorphic coefficients. The quasi-periodicity property~\eqref{cmul} is valid for $k=1,2$ and readily follows inductively for arbitrary~$k$. Likewise, the evenness property~\eqref{cjeven} follows inductively. The coefficient~$c_k^{(k)}(z)$ of a product $A(f_1)\cdots A(f_k)$ is of the form
\be
f_1(z)f_2(z+\eta)\cdots f_k(z+(k-1)\eta),
\ee
so it is in $V_k(\mu)$. The assertion~\eqref{ckV} is then clear from linearity, cf.~Prop.~2.1. The inclusion~\eqref{subs} follows inductively from the constants being a subspace of $\cV_2(\mu)$, cf.~Lemma~3.3.

Next assume $\eta$ satisfies~\eqref{etareq}. Since the generators $A(f)$ have the invariance property~\eqref{AMt} (as proved in Lemma~3.1), this is also true for $A^{(k)}$~\eqref{AVk}, so the properties~(i)--(iii) are a consequence of~Lemma~3.2. Also, holomorphy of the functions~\eqref{poles} follows upon inspection of the locations of the simple poles.
\end{proof}

We are now in the position to characterize the A$\De$Os in the representation $\cD(\mu)$ with $\eta$ satisfying~\eqref{etareq}. We say that a difference operator of the form
\be\label{ASk}
A=\sum_{j=-k}^{k}c_j(z)\exp(j\eta \pz),\ \ \ k>0,\ \ \ \eta\notin\Q\Lambda,
\ee
is an A$\De$O of Sklyanin type if and only if the coefficients are meromorphic functions that satisfy~\eqref{cmul}, are such that the functions~\eqref{poles} are holomorphic,
and have the properties (i)--(iii) detailed in Lemma~3.2. (Since \eqref{i} with $t>0$ is an easy consequence of~\eqref{cmul} and \eqref{i} with $t=0$, this definition is equivalent to the one given in the Introduction, cf.~the paragraph containing~\eqref{cjeven}.) We have already shown that for $\eta$ satisfying~\eqref{etareq} any $A\in\cD(\mu)$ is of Sklyanin type, and we shall now prove the converse.  

We can write any $A$ as a sum of two A$\De$Os containing the even and odd  powers of the shift $\exp(\eta\partial_z)$. We shall call these summands even and odd A$\De$Os, resp. Obviously, $A$ is of Sklyanin type if and only if its even and odd summands are.

\begin{lemma}
Assume $A$ is an A$\De$O of Sklyanin type that is even or odd. Let $k$ be the smallest integer such that $A$ can be written as~\eqref{ASk}. Then we have
\be\label{AcV}
A\in \cV_k(\mu).
\ee
\end{lemma}
\begin{proof}
The restriction on $k$ implies $c_k(z)$ does not vanish identically. Combining~\eqref{cmul}, holomorphy of~\eqref{poles} with $j=k$, and~\eqref{ii}--\eqref{iiia}, we  deduce that $c_k$ belongs to $V_k(\mu)$. For $k=1$ this entails $A$ equals $A(c_1)$, so~\eqref{AcV} is clear.  Next let $k>1$. Recalling Lemma~3.4, it readily follows that $c_k(z)$ can be written as a linear combination of functions of the form
\be
f_1(z)f_2(z+\eta)\ldots f_k(z+(k-1)\eta),\ \ \ f_1,\ldots, f_k\in V_1(\mu).
\ee
Subtracting the associated linear combination of monomials $A(f_1)\ldots A(f_k)$ from $A$, we obtain a constant for the special case $k=2$, so~\eqref{AcV} follows again. 

For $k>2$ the difference is an A$\De$O $A'$ of the form~\eqref{ASk} with $k\to k-2$. Now Lemma~3.2 implies that $A$ has the invariance property~\eqref{AMt}, and by Lemma~3.1 the monomials also satisfy~\eqref{AMt}. Therefore $A'$ satisfies~\eqref{AMt}, so by Lemma~3.2 its coefficients $c_j'(z)$ have the properties (i)--(iii). From this it easily follows that the functions
\be
c_j'(z)\prod_{\ell=-k+3}^{k-3}\theta(2(z+\ell\eta)),\ \ \ |j|\le k-2,
\ee
are holomorphic.
Obviously the coefficients of $A'$  also have the quasi-periodicity property~\eqref{cmul}, so $A'$ is an A$\De$O of Sklyanin type. Recalling~\eqref{subs}, the lemma now follows by finite induction.  
\end{proof}

In view of these lemmas, the following theorem needs no further proof. 
\begin{theorem}
An A$\De$O is of Sklyanin type if and only if it belongs to the representation~$\cD(\mu)$ of the Sklyanin algebra $\cS_{\eta}$ with $\eta\notin\Q\Lambda$. 
\end{theorem}

It should be noted that the lemmas give more information than the theorem. In particular, Lemmas~3.3--3.5 involve far weaker $\eta$-restrictions. 

Thus far, we have viewed~$\eta$ as a fixed parameter. For the Sklyanin generators~$D_0,D_1,D_2$ and~$D_3$ given by~\eqref{D0}--\eqref{D3}, a specific holomorphic dependence on~$\eta$ occurs already in the coefficients. By contrast, we work with generators~$A_R$ with coefficient functions varying over the space~$V_1(\mu)$, which has no dependence on~$\eta$, cf.~\eqref{AR} and the paragraph containing~\eqref{fR}. (We shall elaborate on this difference in Appendix~B.) Accordingly, the $\eta$-dependence in the~$k+1$ coefficients~$c^{(k)}_{k-2m}(\eta;z)$  of the general A$\De$O~$A^{(k)}$~\eqref{AVk} derives solely from the~$\eta$-dependent shifts in the generators~$A_R$. 

We conclude this section by deriving some results pertaining to the~$\eta$-dependence of the latter coefficients. First of all, it follows inductively that they are meromorphic functions of $\eta$ and $z$, with polar divisors contained in the union of hyperplanes
\be\label{polediv}
\cP_k \equiv \{ z+ \ell \eta \in \Lambda/2 \mid |\ell|\le k-1\}.
\ee
Next, we introduce discrete subsets of the complex plane by
\be\label{cDN}
\cD_N\equiv \cup_{\ell=1}^N\frac{\Lambda}{2\ell}.
\ee
Clearly, the hyperplanes in~\eqref{polediv} are distinct, provided $\eta$ does not belong to $\cD_{2k-2}$. When we fix $\eta_0\notin \cD_{2k-2}$, therefore, the coefficients~$c^{(k)}_{k-2m}(\eta_0;z)$ are meromorphic functions of $z$ whose poles can only occur at $4(2k-1)$ distinct locations in a period cell.

To continue, we fix $k>1$ (the case $k=1$ being trivial) and introduce the product functions  
\be\label{Pkm}
P^{(k)}_m(\eta,z)\equiv \prod_{{\substack{\ell =-m\\ \ell\ne k-2m}}}^{k-m}\theta(2z+2\ell\eta),\ \ \ m=0,1,\ldots,k,
\ee
which are obviously holomorphic in $\eta$ and $z$. Our aim is now to show that the functions
\be\label{Hkm}
H^{(k)}_m(\eta,z)\equiv P^{(k)}_m(\eta,z)c^{(k)}_{k-2m}(\eta;z),\ \ \ 0\le m\le k,
\ee
are also holomorphic in $\eta$ and $z$. To be sure, for the special cases $m=0$ and $m=k$ this feature is nearly immediate, and indeed we have already encountered the $m=0$ product function in a related setting, cf.~\eqref{defPk}. For the general case, however, it is already an arduous task to verify holomorphy directly for small~$k$-values and~$m$ near~$k/2$, as there are delicate cancellations present.

Our general holomorphy proof involves several steps.
First, we observe that when we fix $\eta_0\notin \Q\Lambda$, the  theta function product $P^{(k)}_m(\eta_0,z)$ has~$4k$
simple zeros in a period cell, and it follows from Lemma~3.2 that the zero locations match the locations of the simple poles of~$c^{(k)}_{k-2m}(\eta_0;z)$. Hence the functions $H^{(k)}_m(\eta_0,z)$ are holomorphic in~$z$.

Now it suffices to require $\eta_0\notin \cD_k$  for the zeros of $P^{(k)}_m(\eta_0,z)$ to be simple. As shall become clear shortly, this also suffices to retain simple poles for~$c^{(k)}_{k-2m}(\eta_0;z)$ at the corresponding locations, so that $H^{(k)}_m(\eta_0,z)$ remains holomorphic. But it is expedient to first prove holomorphy for a fixed $\eta_0$ that does not belong to the larger discrete set $\cD_{2k-2}$. From Lemma~3.2 we already know holomorphy when $\eta_0$ does not belong to the dense set $\Q\Lambda$, so we now fix $\eta_0$ satisfying
\be
\eta_0\in\Q\Lambda ,\ \ \ \eta_0\notin \cD_{2k-2},
\ee
until further notice.

Next, we observe that the complement of $\Q\Lambda$ is dense as well. Thus we can find a sequence $\eta_n$ satisfying
\be\label{etan}
\eta_n\notin\Q\Lambda,\ \ \ \lim_{n\to\infty}\eta_n=\eta_0.
\ee
The functions $H^{(k)}_m(\eta_n,z)$ are holomorphic, and they converge to $H^{(k)}_m(\eta_0,z)$ uniformly on $z$-compacts that are disjoint from the polar divisor $D_0$ given by~\eqref{polediv} with $\eta=\eta_0$. Choosing $z_0\in D_0$, we now aim to show that~$H^{(k)}_m(\eta_0,z)$ has no pole at $z=z_0$, so as to deduce holomorphy of~$H^{(k)}_m(\eta_0,z)$. 

To this end we choose $r>0$ small enough so that the punctured disc $|z-z_0| \in (0,r]$ does not meet~$D_0$. Since $\eta_n\notin\cD_{2k-2}$, the polar divisors $D_n$ given by~\eqref{polediv} with $\eta=\eta_n$ meet the disc $|z-z_0|\le r$ in a unique $z_n$ for $n$ large enough, with $z_n\to z_0$ as $n\to\infty$. Therefore, we have 
\be
\lim_{n\to\infty}H^{(k)}_m(\eta_n,z)=H^{(k)}_m(\eta_0,z),\ \ \ |z-z_0|\in(0,r],
\ee
uniformly on the circle $|z-z_0|=r/2$. Invoking the following elementary lemma, it now follows that~$H^{(k)}_m(\eta_0,z)$ has no pole at $z=z_0$ and holomorphy of~$H^{(k)}_m(\eta_0,z)$ results.

\begin{lemma}
Assume $F_n(w)$ is a sequence of functions that are holomorphic for $|w|\le r$ and $F_0(w)$ is holomorphic for $|w|\in(0,r]$. Next, assume
\be
\lim_{n\to\infty}F_n(w)=F_0(w),\ \ \ |w|\in(0,r],
\ee
uniformly on the circle $|w|=r/2$. Then the function $F_0(w)$ is holomorphic at $w=0$.
\end{lemma}
\begin{proof}
For $|w|<r/2$ we have the Cauchy integral formula
\be
F_n(w)=\frac{1}{2\pi i}\oint_{|v|=r/2}\frac{F_n(v)}{v-w}dv.
\ee
For $n\to\infty$ the right-hand side converges to the limit function
\be
L(w)=\frac{1}{2\pi i}\oint_{|v|=r/2}\frac{F_0(v)}{v-w}dv,
\ee
which is holomorphic for $|w|<r/2$. The left-hand side converges to $F_0(w)$ for~$|w|\in(0,r/2)$, so $F_0(w)=L(w)$ and holomorphy at $w=0$ follows.
\end{proof}

The upshot of our reasoning is that the functions $H^{(k)}_m(\eta,z)$ are holomorphic in~$\eta$ and~$z$, unless $\eta$ belongs to the discrete set $\cD_{2k-2}$ and in addition $z$ belongs to the discrete set obtained from the polar divisor $\cP_k$~\eqref{polediv} by fixing $\eta\in\cD_{2k-2}$. This exceptional set is a discrete subset of $\C^2$, so we are now in the position to invoke Hartogs' theorem on analytic completion to conclude that the functions~$H^{(k)}_m(\eta,z)$ are in fact jointly holomorphic in~$(\eta,z)$. (For the case at hand, Hartogs' theorem can be proved by a second application of the Cauchy integral formula, cf.~\cite{GH78}, p.~7.)

We now summarize and extend this finding in the following theorem.

\begin{theorem}
Let $k>1$ and $m\in\{0,\ldots,k\}$. Then the functions $H^{(k)}_m(\eta,z)$ defined by~\eqref{Pkm}--\eqref{Hkm} are holomorphic in $\C^2$. Fixing~$\eta\notin\cD_k$, the coefficients~$c^{(k)}_{k-2m}(\eta;z)$ have at most simple poles for $z+\ell\eta\in\Lambda/2$, where $\ell=-m,\ldots, k-m$ and $\ell\ne k-2m$, with residues~$r_{k-2m}(t,\ell)$ at $z=\omega_t-\ell \eta$ that are related by~\eqref{iiib}. Finally, fixing $z\notin\Lambda/2$, we have 
\be\label{etaper}
c^{(k)}_{k-2m}(\eta+1;z)=c^{(k)}_{k-2m}(\eta;z),\ \ \ c^{(k)}_{k-2m}(\eta+\tau;z)=\mu^{e(k,m)}c^{(k)}_{k-2m}(\eta;z),
\ee
where the exponent is the integer
\be\label{ekm}
e(k,m)=(k-2m)^2/2-k/2.
\ee
\end{theorem}
\begin{proof}
We have already proved the holomorphy assertion. From this and the simplicity of the zeros of~$P^{(k)}_m(\eta,z)$ for $\eta\notin\cD_k$, the assertions concerning the pole locations and their simplicity follow. Furthermore, the residue relations hold true for $\eta\notin\Q\Lambda$, and the residues are finite and continuous in $\eta$ as long as $\eta$ stays away from the set $\cD_k$ (on which the pole multiplicity can be greater than one). Hence the residue relations continue to hold on the complement of~$\cD_k$. 

It remains to prove the quasi-periodicity claim. Of course, 1-periodicity is plain. More generally, it is straightforward to check that~\eqref{etaper}--\eqref{ekm} are valid for small~$k$, and for general~$k$ when~$m$ equals~$0$ and~$k$. 

Turning to the general case, we first express the coefficients of the product $A(a,\nu)A^{(k)}$, where $A(a,\nu)$ is given by~\eqref{Apar} and $A^{(k)}$ is the general A$\De$O~\eqref{AVk} in~$\cV_k(\mu)$, in terms of those of~$A^{(k)}$. This yields
\be\label{crec}
c_{k+1-2n}^{(k+1)}(\eta;z)=\frac{\theta(z+\vec{a}-\nu)}{\theta(2z)}c_{k-2n}^{(k)}(\eta;z+\eta)+
\frac{\theta(z-\vec{a}+\nu)}{\theta(2z)}c_{k-2n+2}^{(k)}(\eta;z-\eta),
\ee
where $n=0,\ldots,k+1$, and $c^{(k)}_{\pm (k+2)}\equiv 0$. 

We now proceed by induction on~$k$. It is easy to check quasi-periodicity with the exponents~\eqref{ekm} for~$k$ equal to 1 and 2. Assuming quasi-periodicity is valid with the stated exponents up to and including~$k\ge 2$, we exploit the recurrence~\eqref{crec} to verify its validity for~$k+1$, as follows. Using quasi-periodicity in~$z$ (as given by~\eqref{cmul}), we have
\be
c_{k-2n}^{(k)}(\eta+\tau;z+\eta+\tau)=\mu^{k-2n}c_{k-2n}^{(k)}(\eta+\tau;z+\eta).
\ee
Next, we use the induction assumption to obtain
\be
c_{k-2n}^{(k)}(\eta+\tau;z+\eta+\tau)=\mu^{k-2n}\mu^{e(k,n)}c_{k-2n}^{(k)}(\eta;z+\eta).
\ee
Thus we should check
\be
k-2n+e(k,n)=e(k+1,n),
\ee
which is routine. Likewise, the second coefficient in~\eqref{crec} leads to the easily verified identity
\be
-k+2n-2+e(k,n-1)=e(k+1,n).
\ee
Hence quasi-periodicity of~$c_{k+1-2n}^{(k+1)}(\eta;z)$ with the asserted exponent is now clear from~\eqref{crec}.
\end{proof}

To conclude this section, we would like to repeat that Theorem~3.9 is concerned with the $\eta$-dependence of the coefficients of A$\De$Os arising from generators~$A(f)$~\eqref{Af} that only depend on $\eta$ via the shifts. In this connection we recall that we found it convenient to allow $\eta$-dependence of the four functions~\eqref{fag}--\eqref{fpagp} in the proof of Lemma~3.3, so as to show that $\cV_2(\mu)$ contains the constants. In view of this extra dependence,  this conclusion only pertains to a fixed $\eta$-value. In particular, it does not entail that the above coefficients~$c^{(2)}_0(\eta;z)$ can ever be constant in~$z$ and nonconstant in~$\eta$. Indeed, such a behavior can be ruled out, thanks to the holomorphy and quasi-periodicity results of the theorem. Rather, Lemma~3.3 implies that for any fixed nonzero~$\eta_0$ there exist coefficients~$c^{(2)}_0(\eta;z)$ that reduce to a nonzero constant for $\eta=\eta_0$ (together with $c^{(2)}_{\pm 2}(\eta_0;z)$ vanishing identically).


\section{A$\De$Os of van Diejen type}

\subsection{The Sklyanin algebras $\cS_{\pm}$}

Thus far the shift constant $\eta$ and elliptic modulus $\tau$ have played completely different roles. In this section we switch to notation in which the parameters
\be\label{apm}
a_+\equiv -i\tau,\ \ \ a_-\equiv -2i\eta,\ \ \ \im \tau>0,\ \ \im \eta>0,
\ee
enter in a symmetric way. This is because the results in this section involve the elliptic gamma function $G(r,a_+,a_-;z)$ introduced and studied in~\cite{Rui97}, and because the symmetries of the van Diejen A$\De$Os~\cite{vDie94} find their most natural expression in this notation~\cite{Rui04}. Note in particular that in this section we require $\eta$ to be in the upper half plane, alongside with $\tau$.

Taking $r=\pi$ henceforth, the elliptic gamma function can be defined by
\be\label{Gell}
G(z)=\prod_{m,n=1}^{\infty}
\frac{1-q_{+}^{2m-1}q_{-}^{2n-1}e^{-2i\pi z}}{1-q_{+}^{2m-1}
q_{-}^{2n-1}e^{2i\pi z}},  \ \ \ q_{\pm}\equiv \exp (-\pi a_{\pm}),\ \ \ \re a_+,\re a_->0.
\ee
(Here and below, we suppress the dependence on the parameters when no confusion can arise.) Obviously, $G(z)$ is a meromorphic function that is symmetric under the interchange of $a_+$ and $a_-$, a property that is now often referred to as modular invariance. Moreover, $G(z)$ satisfies the A$\De$Es (analytic difference equations)
\be\label{Gades}
G(z+ia_{-\de}/2)/G(z-ia_{-\de}/2) =R_{\de}(z),\ \ \ \ \de=+,-,
\ee
where the right-hand side functions are given by
\be
R_{\de}(z)=\prod_{k=1}^{\infty}(
1-q_{\de}^{2k-1}\exp (2i\pi z)) (
1-q_{\de}^{2k-1}\exp (-2i\pi z)).
\ee
Hence $R_{\de}(z)$ is a holomorphic even 1-periodic function that satisfies the A$\De$E
\be\label{Rade}
R_{\de}(z+ia_{\de}/2)/R_{\de}(z-ia_{\de}/2)=-\exp(-2i\pi z).
\ee

The relation of the functions $R_{\pm}(z)$ to the theta functions used so far (cf.~\eqref{th1}--\eqref{th24}) is given by
\be
R_{\de}(z)=\prod_{k=1}^{\infty}(1-q_{\de}^{2k})^{-1} \cdot \theta_4(z|ia_{\de}),\ \ \ \de=+,-,
\ee
\be\label{thR}
\theta (z|\tau)=iq^{1/4}\prod_{k=1}^{\infty}(1-q^{2k})\cdot \exp(-i\pi z)R_{+}(z-\tau/2).
\ee
The elliptic gamma function satisfies a variety of multiplication formulae~\cite{Rui97}; in particular, its duplication formula entails
\be\label{Rdup}
R_{\de}(2z)=R_{\de}(z\pm ia_{\de}/4,z+1/2\pm ia_{\de}/4),
\ee
which is the counterpart of~\eqref{dup}. 

Consider now the pair of A$\De$Os
\be\label{AR2}
A_{R,\de}(l;z)\equiv f_{\de}(l;z)\exp(ia_{-\de}\partial_z/2)+(z\to -z),\ \ \ l\in\C^4,\ \ \ \de=+,-,
\ee
where
\be\label{fde}
f_{\de}(l;z)\equiv R_{\de}(z+\vec{l}+ia_{-\de}/4)/R_{\de}(2z-ia_{\de}/2),\ \ \ l\in\C^4,\ \ \ \de=+,-.
\ee
Using~\eqref{apm} and~\eqref{thR}, we see that $A_{R,+}(l;z)$ is of the form~\eqref{AR}, with the multiplier given by
\be\label{fmu}
f_{+}(l;z+ia_{+})/f_{+}(l;z)=\exp\Big(-2\pi i \Big(\sum_{n=1}^4l_n+2ia\Big)\Big)=:\mu ,
\ee
cf.~\eqref{Rade}. Here, we have introduced the modular invariant parameter
\be\label{defa}
a\equiv (a_{+}+a_-)/2.
\ee
It should be noted, however, that we have introduced an explicit dependence of the coefficients on $a_{-}$. This is expedient with an eye on the relation to the van Diejen A$\De$Os; we need only replace $l_n$ by $l_n-ia/2$ in $A_{R,+}(l;z)$ to revert to an A$\De$O of the form $A(a,\nu)$~\eqref{Apar}, up to a multiplicative constant~$\exp(- 4i\pi\nu)$ following from~\eqref{thR}.

As a consequence, we have rewritten the generators of the $\cD(\mu)$-representation of the Sklyanin algebra $\cS_{\eta}$ with respect to the elliptic lattice $\Z+\tau \Z$ in such a way that when the parameters $a_+$ and $a_-$ are interchanged, we get the generators of the $\cD(\mu)$-representation of the Sklyanin algebra $\cS_{\tau/2}$ with respect to the elliptic lattice $\Z+2\eta \Z$. Equivalently, letting
\be
\Lambda_{\de}\equiv \Z +ia_{\de}\Z,\ \ \ \de=+,-,
\ee
we get generators of the $\cD(\mu)$-representation of two Sklyanin algebras $\cS_+$ and $\cS_-$ associated to elliptic lattices $\Lambda_+$ and $\Lambda_-$, respectively. (Such a modular pair of Sklyanin algebras was first obtained and studied by Spiridonov~\cite{Spi09}.)

\subsection{Kernel identities for the Sklyanin generators}

We proceed to derive kernel identities that hold for both sets of generators at once. Specifically, we introduce a kernel function
\be\label{defK}
K(\gamma;z,y)\equiv G(\pm z\pm y-\gamma),\ \ \  \gamma\in \C,
\ee
and consider the question if and when identities of the form
\be\label{Kid}
A_{R,\de}(l;z)K(\gamma;z,y)=c_{\de} A_{R,\de}(k;y)K(\gamma;z,y),\ \ \ c_{\de}\in\C,\ \ \ \de=+,-,
\ee
are valid. More precisely, we ask: Fixing  $\gamma\in\C$ with
\be\label{gass}
2\gamma\notin \Lambda_+\cup \Lambda_-,
\ee
until further notice, can one find $c_+,c_-\in\C$ and $l,k\in\C^4$ such that~\eqref{Kid} holds true? 

To answer this, we divide~\eqref{Kid} by $K(\gamma;z-ia_{-\de}/2,y)$ and then use the $G$-A$\De$Es~\eqref{Gades} to obtain 
\bea\label{Rreq}
&  &  \frac{R_{\de}(z+\vec{l}+ia_{-\de}/4)}{R_{\de}(2z-ia_{\de}/2)}\cdot\frac{R_{\de}(z\pm y-\gamma)}{R_{\de}(z\pm y+\gamma)}+\frac{R_{\de}(z-\vec{l}-ia_{-\de}/4)}{R_{\de}(2z+ia_{\de}/2)}   
\nonumber \\
  & = & c_{\de}\left( \frac{R_{\de}(y+\vec{k}+ia_{-\de}/4)}{R_{\de}(2y-ia_{\de}/2)} 
\cdot\frac{R_{\de}(z+ y-\gamma)}{R_{\de}(z+ y+\gamma)} +
\Big( y\to -y\Big)\right).
\eea
For this to hold, it is necessary that the residues at the (generically) simple pole $z=y-\gamma+ia_{\de}/2$ of the left-hand and right-hand sides be equal. 
This yields upon canceling equal factors (in particular a constant factor $R_{\de}(ia_{\de}/2-2\gamma)$, which does not vanish thanks to our assumption~\eqref{gass}), 
\be\label{RcR}
R_{\de}(y-\gamma +ia_{\de}/2 +\vec{l}+ia_{-\de}/4)=c_{\de}
R_{\de}(y-\vec{k}-ia_{-\de}/4).
\ee
Clearly, \eqref{RcR} implies
\be\label{c=1}
c_{\de}=1,
\ee
\be\label{klg}
k_n=-l_n+\gamma -ia,\ \ \ n=1,2,3,4,
\ee
with the last equality holding modulo permutations and addition of integers. 

From now on, we require~\eqref{c=1}--\eqref{klg}, so that we get equal residues at the pole $z=y-\gamma+ia_{\de}/2$. To obtain residue equality at all poles $\Lambda_{\de}$-congruent to the latter, we must next ensure that the first term on the left-hand side of~\eqref{Rreq} and the second term on the right-hand side have equal multipliers under $\Lambda_{\de}$-translations. Taking $z\to  z+ia_{\de}$ and using~\eqref{Rade}, this requirement yields
\be\label{gamreq}
\exp(4\pi i\gamma )=\exp\Big( 2\pi i\Big(\sum_{n=1}^4l_n+2ia\Big)\Big)=1/\mu,
\ee
cf.~\eqref{fmu}.

Imposing~\eqref{gamreq} as well, all terms have equal multipliers under $\Lambda_{\de}$-translations of~$z$. Residue equality at the pole $z=-y-\gamma+ia_{\de}/2$ is readily checked, so the difference of left-hand and right-hand side of~\eqref{Rreq} yields a function $d(z)$ that has period 1, quasi-period~$ia_{\de}$ with multiplier $1/\mu$, and no $y$-dependent $z$-poles. 

We proceed to show that $d(z)$ vanishes identically. 
By virtue of our assumption~\eqref{gass} we have
 $\mu\ne \exp(2\pi ka_{\de})$,  $\forall k\in\Z$. Hence one need only check that $d(z)$ has no poles for $\pm z=0,1/2,ia_{\de}/2 +1/2,ia_{\de}/2$.  Using evenness of $R_{\de}(z)$ and~\eqref{Rade}, this follows by straightforward calculations that we omit. We now summarize and extend our findings in the following theorem.
 
\begin{theorem}
Let $\gamma,a_+,a_-\in\C$ with $\re a_+,\re a_->0$, and assume $l\in\C^4$ satisfies~\eqref{gamreq}. Then we have
\be\label{Kid2}
A_{R,\de}(l;z)K(\gamma;z,y)=A_{R,\de}(k;y)K(\gamma;z,y),\ \ \ \de=+,-,
\ee
where $K(\gamma;z,y)$ is given by~\eqref{defK} and $k$ by~\eqref{klg}.
\end{theorem}
\begin{proof}
We have already shown that~\eqref{Kid2} is valid under the assumption~\eqref{gass}. Since we are dealing with functions that are meromorphic in $\gamma$, the theorem readily follows.
\end{proof}

From~\eqref{klg}--\eqref{gamreq} and~\eqref{fmu} we obtain
\be
f_{\de}(k;z+ia_{\de})/f_{\de}(k;z)=f_{\de}(l;z+ia_{\de})/f_{\de}(l;z)=\mu,\ \ \ \de=+,-.
\ee
As a consequence, if the A$\De$O $A_{R,\de}(l;z)$ in~\eqref{Kid2} belongs to the representation~$\cD(\mu)$, then so does~$A_{R,\de}(k;z)$.
 We proceed to show that there are actually two distinct bijections involved. To this end we introduce 
\be\label{defkap}
\kappa := (1,1,1,1),
\ee
and two maps
\be\label{phi1}
\phi_1\colon \C^4\to\C^4,\ \ \ l\mapsto \hat{l}:= -l+ \langle \kappa,l\rangle \kappa/2,
\ee
\be\label{phi2}
\phi_2\colon \C^4\to\C^4,\ \ \ l\mapsto \tilde{l}:= \hat{l}+\kappa/2,
\ee
where $\langle \cdot,\cdot\rangle$ denotes the standard bilinear form on~$\C^4$. Clearly, $\phi_1$ is a reflection, whereas
\be
\phi_2(\phi_2(l))=l+\kappa.
\ee
We now set (cf.~\eqref{gamreq})
\be\label{gp}
2\hat{\gamma}\equiv 2ia+ \sum_{n=1}^4l_n  \pmod 1,\ \ \ \re \hat{\gamma}\in [0,1/2),
\ee
\be\label{gt}
\tilde{\gamma}:= \hat{\gamma}+1/2,
\ee
and introduce kernel functions
\be\label{K12}
K_1(z,y):= K(\hat{\gamma};z,y),\ \ \ K_2(z,y):= K(\tilde{\gamma};z,y).
\ee
Then the theorem has the following corollary.

\begin{corollary}
Let $a_+,a_-\in\C$ with $\re a_+,\re a_->0$, and let $l\in\C^4$. Then we have
\be\label{Kidj}
A_{R,\de}(l;z)K_j(z,y)=A_{R,\de}(\phi_j(l);y)K_j(z,y),\ \ \ \de=+,-,\ \ \ j=1,2.
\ee
\end{corollary} 
\begin{proof}
For both choices of $j$, the condition~\eqref{gamreq} is satisfied, and $\phi_j(l)$ amounts to~\eqref{klg}. Hence the identities~\eqref{Kid2} imply~\eqref{Kidj}.
\end{proof}

It is clear from the above that the kernel functions $K_j(z,y)$ give rise to two involutions
\be\label{Phi_j}
\Phi_j\colon A_{R,\pm}(l;z)\mapsto A_{R,\pm}(\phi_j(l);z),\ \ \ j=1,2,
\ee
of the generators of the Sklyanin algebras~$\cS_{\pm}$. Also, a moment's thought shows that both bijections preserve the Sklyanin relations for twofold products  in the opposite order (see also Appendix~B, in particular~\eqref{aa1}--\eqref{aa2}). Taking linear combinations of arbitrary products of generators, it therefore follows that the two maps extend to order-reversing involutive automorphisms of the $\cD(\mu)$-representations of~$\cS_{\pm}$.

Thus far we have viewed the $\cD(\mu)$-representation of the Sklyanin algebras $\cS_{\pm}$ as being defined in terms of endomorphisms on the space $\cM_e$ of even meromorphic functions. As a corollary of our kernel results, however, we can easily show that for special $\mu$ finite-dimensional submodules occur.

The simplest case is $\mu=1$. Then all coefficient functions are elliptic functions, and since we have $\hat{\gamma}=0$ and $\tilde{\gamma}=1/2$, we get
\be
K_j(z,y)=1,\ \ \ j=1,2,\ \ \ (\mu=1).
\ee
Thus \eqref{Kidj} reduces to the identities
\be
f_{\de}(l;z)+(z\to -z)=f_{\de}(\phi_j(l);y)+(y\to -y),\ \ \ j=1,2\ \ \ \ (\mu=1),
\ee
and it follows in particular that the algebras leave the space of constant functions invariant.

Next, we study the choice
\be
\mu =\exp(-2\pi a_{-\de})\Rightarrow 2\gamma =-ia_{-\de} \pmod 1,
\ee
cf.~\eqref{gamreq}. From~\eqref{Gades} we deduce that the two kernel functions reduce to
\be
R_{\de}(z+ y)R_{\de}(z- y),\ \ \ R_{\de}(z+ y +1/2)R_{\de}(z- y +1/2).
\ee
Viewing them as functions of~$z$ depending on a parameter $y$, their span as~$y$ varies over~$\C$ yields a two-dimensional subspace~$\Theta_{1,\de}$ of $\cM_e$, namely the space of even, order-2 theta functions (with respect to $\Lambda_{\de}$). By virtue of the kernel identities, the action of the A$\De$Os $A_{R,\pm}(l;z)$ on these two kernel functions can be translated to an action on the parameter~$y$, and so it follows that we have
\be
A_{R,\pm}(l;z)\Theta_{1,\de}\subset \Theta_{1,\de}\Rightarrow \cS_{\pm}\Theta_{1,\de} \subset \Theta_{1,\de},\ \ \ (\mu=\exp(-2\pi a_{-\de})).
\ee

More generally, the choices
\be
\mu =\exp(-2N\pi a_{-\de})\Rightarrow 2\gamma =-iNa_{-\de} \pmod 1,\ \ \ N=0,1,2,\ldots,
\ee
yield two kernel functions
\be\label{kerN}
K_{\de}^{(N)}(z,y)\equiv\prod_{k=0}^{N-1}R_{\de}(z\pm y +(2k-N+1)ia_{-\de}/2),\ \ \ K_{\de}^{(N)}(z+1/2,y),
\ee
which, viewed as functions of $z$ depending on a parameter $y\in\C$, belong to the $(N+1)$-dimensional space $\Theta_{N,\de}$ of even, order-$2N$ theta functions. It then follows as before that whenever these functions span this space, we have
\be
\cS_{\pm}\Theta_{N,\de} \subset \Theta_{N,\de},\ \ \ (\mu=\exp(-2N\pi a_{-\de})).
\ee

As a consequence of the kernel identities, we have therefore arrived at the finite-dimensional modules first studied by Sklyanin~\cite{Skl83}. Moreover, the second kernel function in~\eqref{kerN} is proportional to the reproducing kernel for the module~$\Theta_{N,\de}$. The explicit form of this reproducing kernel was already conjectured by Sklyanin, but first proved by Rosengren~\cite{Ros04}. To be sure, in these papers only the algebra $\cS_{\eta}\sim\cS_{+}$ and modules~$\Theta_{N,+}$ are considered, with further restrictions on $a_+$ and $a_-$ (as given by~\eqref{apm}).

The insight that the latter modules are also left invariant by the `modular copy' $\cS_{-}$ seems to be new. It would be worthwhile to study the action of $\cS_{-}$ on $\Theta_{N,+}$ in more detail, but this is beyond our scope.

A different type of finite-dimensional module for~$\cS_{\pm}$ arises for the choices
\be
\mu =\exp(-2\pi (N_{+}a_{+}+N_{-}a_{-}))\Rightarrow 2\gamma =-i(N_{+}a_{+}+N_{-}a_{-}) \pmod 1,\ \ \ N_{+},N_{-}\ge 1.
\ee
Indeed, a straightforward calculation using the $G$-A$\De$Es~\eqref{Gades} yields the two kernel functions
\be\label{kerNpNm}
K^{(N_{+},N_{-})}(z,y)\equiv K_{-}^{(N_{+})}(z+iN_{-}a_{-}/2,y)
K_{+}^{(N_{-})}(z-iN_{+}a_{+}/2,y),\ \ \ K^{(N_{+},N_{-})}(z+1/2,y).
\ee
The first and second factor admit an expansion as tensors of rank at most $N_++1$ and $N_-+1$, resp., so upon multiplying out, we see that we obtain a module of dimension at most $(N_{+}+1)(N_{-}+1)$. Once again, we do not embark on a further study.

\subsection{The connection to the van Diejen A$\De$Os}

We proceed to clarify the connection of the products
\be\label{ARlm}
A_{R,\de}^{(2)}(l,m;z)\equiv A_{R,\de}(l;z)A_{R,\de}(m;z),
\ \ \ \langle \kappa, l-m\rangle =0,\ \ \ \de=+,-,
\ee
and the kernel identities
\be\label{Kidj2}
A_{R,\de}^{(2)}(l,m;z)K_j(z,y)=A_{R,\de}^{(2)}(\phi_j(m),\phi_j(l);y)K_j(z,y),\ \ \ j=1,2,
\ee
(which are immediate consequences of~\eqref{Kidj}), to the van Diejen A$\De$Os and their kernel identities, as encoded in~Proposition~3.1 of~\cite{Rui09}.

To begin with, we should define the van Diejen A$\De$Os. With the real period $\pi/r$ of their coefficients fixed to 1 (as in the elliptic gamma function~\eqref{Gell}), they depend on the modular parameters $a_+,a_-$ and a `coupling' vector $h\in\C^8$. Specifically, they are given by
\be\label{AD}
A_{D,\de}(h;z)\equiv V_{\de}(h;z)\exp(-ia_{-\de}\partial_z)+V_{\de}(h;-z)\exp(ia_{-\de}\partial_z)+V_{b,\de}(h;z),
\ee
where
\be\label{Ve}
V_{\de}(h;z)\equiv\frac{\prod_{n=1}^8R_{\de}(z-h_n-ia_{-\de}/2)}{R_{\de}(2z+ia_{\de}/2)R_{\de}(2z-ia_{-\de}+ia_{\de}/2)},
\ee
\be\label{Vbe}
V_{b,\de}(h;z)\equiv \frac{\sum_{t=0}^3p_{t,\de}(h)
[\cE_{t,\de}(\xi;z)-\cE_{t,\de}(\xi;\omega_{t,\de})]}{2R_{\de}(\xi -ia_{\de}/2)R_{\de}(\xi -ia_{-\de}-ia_{\de}/2)}.
\ee
Here we are using half-periods
\be
\omega_{0,\de}=0,\ \ \omega_{1,\de}=1/2,\ \ \omega_{2,\de}=ia_{\de}/2,\ \ \omega_{3,\de}=-1/2-ia_{\de}/2,
\ee
the product functions are given by
\be\label{p01}
p_{0,\de}(h)\equiv \prod_n R_{\de}(h_n),\ \ \ p_{1,\de}(h)\equiv \prod_n R_{\de}(h_n-1/2),
\ee
\be
p_{2,\de}(h)\equiv \exp(-2\pi a_{\de})\prod_n \exp(-i\pi h_n)R_{\de}(h_n-ia_{\de}/2),
\ee
\be\label{p3}
p_{3,\de}(h)\equiv \exp(-2\pi a_{\de})\prod_n \exp(i\pi h_n)R_{\de}(h_n+1/2+ia_{\de}/2),
\ee
and $\cE_{t,\de}$ reads
\be\label{cEt}
\cE_{t,\de}(\xi;z)\equiv
\frac{R_{\de}(z+\xi -ia-\omega_{t,\de})R_{\de}(z-\xi +ia-\omega_{t,\de})}{R_{\de}(z-ia-\omega_{t,\de})R_{\de}(z+ia-\omega_{t,\de})},\ \ \
t=0,1,2,3.
\ee

Clearly, the functions $\cE_{t,\de}(\xi;z)$ and $V_{b,\de}(h;z)$ are elliptic in $z$ with periods 1 and $ia_{\de}$. It is not clear by inspection, but true that the function $V_{b,\de}(h;z)$ does not depend on the parameter $\xi\in\C$. We claim this follows from Lemma~3.2 in~\cite{Rui04}. The function at issue there deviates from $V_{b,\de}(h;z)$ by a constant, but this constant is $\xi$-independent. Indeed, in~\cite{Rui04} (and also in~\cite{Rui09}) the factor $\cE_{t,\de}(\xi;\omega_{t,\de})$ is replaced by $\cE_{t,\de}(\xi;z_t)$, with $z_0=z_2=1/2$ and
$z_1=z_3=0$. Since $z_t$ does not depend on $\xi$, the claim follows. (It can be verified directly by noting that $V_{b,\de}(h;z)$ is elliptic in $\xi$ and checking that residues vanish.)

Next, we note that~\eqref{cEt} implies
\be\label{cEcon}
\cE_{t,\de}(\xi;\omega_{t,\de})=\left(\frac{R_{\de}(\xi -ia)}{R_{\de}(ia)}\right)^2.
\ee
Therefore, the additive constant diverges when 
$R_{\de}(ia)$ vanishes. Accordingly, we assume from now on that $a_{\pm}$ are such that we have
\be\label{apmz}
R_{\pm}(ia)\ne 0,\ \ \ a=(a_++a_-)/2.
\ee

We proceed by noting
\be\label{Vhm}
V_{\de}(h;-(z+ia_{\de}))/V_{\de}(h;-z)=\exp\Big(-2\pi i\Big(\sum_nh_n+4ia\Big)\Big).
\ee
Specializing to $h\in\C^8$ satisfying
\be\label{hmu}
\exp\Big(-2\pi i\Big(\sum_nh_n+4ia\Big)\Big)=\mu^2,
\ee
 it easily follows that $V_{+}(h;-z)$ belongs to the space $V_2(\mu)$ defined above Prop.~2.1. 

To establish the connection with the A$\De$Os of Sklyanin type characterized in Section~3, let us rewrite~\eqref{AD} as
\be\label{ADalt}
A_{D,\de}(h;z)=c_{-2,\de}(z)\exp(-ia_{-\de}\partial_z)+c_{2,\de}(z)\exp(ia_{-\de}\partial_z)+c_{0,\de}(z),
\ee
and use from now on notation that encodes the two Sklyanin algebras $\cS_{\pm}$ at issue in this section. Thus, we see from the previous paragraph that 
\be\label{c2V}
c_{2,\de}(z)\in V_{2,\de}(\mu).
\ee
Obviously, we also have 
\be
c_{-2,\de}(z)=c_{2,\de}(-z),
\ee
and $c_{0,\de}(z)$ is an even elliptic function. Furthermore, it should be noted that for $h$ and $h'$ satisfying~\eqref{hmu} we have
\be
[A_{D,+}(h;z),A_{D,-}(h';z)]_{-}=0,
\ee
in contrast to the non-commutativity of $A_{R,+}(l;z)$ and $A_{R,-}(l';z)$ for $l$ and $l'$ yielding the same multiplier.

Assuming from now on that the modular parameters $a_{\pm}$ satisfy (cf.~\eqref{etareq} and \eqref{apm})
\be\label{apmgen} 
ia_{\de}\notin\Q\Lambda_{-\de},\ \ \ \de=+,-,
\ee
(implying in particular~\eqref{apmz}), it follows that $c_{0,\de}(z)$ has simple poles for $z$-values $\Lambda_{\de}$-congruent to $\omega_{t,\de}\pm ia_{-\de}/2$, and no other poles.

In order to prove that we have
\be\label{ADcV}
A_{D,\de}(h;z)\in \cV_{2,\de}(\mu),
\ee
it is therefore enough to check that the residues of $c_{0,\de}(z)$ are related to those of $c_{2,\de}(z)$ in the way specified in Lemma~3.2. First, setting
\be
\rho_{\de}\equiv {\rm Res}\left( \frac{1}{R_{\de}(z)}\right)_{z=-ia_{\de}/2},
\ee
it is easy to check from~\eqref{Vbe}--\eqref{cEt} that we have
\be\label{c0res}
{\rm Res}\big(c_{0,\de}(z)\big)_{z=\omega_{t,\de}\pm ia_{-\de}/2}=\pm\frac{\rho_{\de}}{2}\frac{p_{t,\de}(h)}{R_{\de}(ia_{\de}/2+ia_{-\de})}.
\ee

Secondly, we need the pertinent residues of $c_{2,\de}(z)$. For $z$ equal to $-ia_{-\de}/2$ and $-ia_{-\de}/2+1/2$ they are easily seen to be equal to minus the residues 
~\eqref{c0res} of $c_{0,\de}(z)$ at these $z$-values. The calculation of the residues for $z$ equal to $-ia_{-\de}/2+ia_{\de}/2$ and $-ia_{-\de}/2+1/2+ia_{\de}/2$ is more arduous, but straightforward. The result is that they equal the residues 
~\eqref{c0res} of $c_{0,\de}(z)$ at these $z$-values times a factor $-\mu$.  

As a consequence, the coefficients of the A$\De$O~$A_{D,\de}(h;z)$ given by \eqref{ADalt} have the properties detailed below~\eqref{ASk}, so it is of Sklyanin type. Furthermore, it is manifestly even, so~Lemma~3.6 implies that \eqref{ADcV} holds true. We summarize this discussion in the following theorem.

\begin{theorem}
Assume $h\in\C^8$ satisfies~\eqref{hmu} and $a_+, a_-$ are numbers in the (open) right half plane satisfying~\eqref{apmgen}. Then the van Diejen A$\De$Os $A_{D,\de}(h;z)$ given by~\eqref{AD}--\eqref{cEt} belong to the $\cD(\mu)$-representation of the Sklyanin algebra $\cS_{\de}$ and satisfy~\eqref{ADcV}.
\end{theorem}
 
Next, we consider the twofold product~\eqref{ARlm},  still assuming~\eqref{apmgen}. Requiring 
\be
\exp\Big(- 2\pi i\Big(\sum_{n=1}^4k_n+2ia\Big)\Big)=\mu,\ \ \ k=l,m,
\ee
it belongs to~$\cV_{2,\de}(\mu)$. The coefficient of $\exp(-ia_{-\de}\partial_z)$ is of the form $V_{\de}(h;z)$~\eqref{Ve}, with
\be\label{hlm}
h_n=l_n-ia_{-\de}/4,\ \ \ h_{n+4}=m_n+ia_{-\de}/4,\ \ \ n=1,2,3,4.
\ee
Hence it follows that we have an equality
\be\label{ARAD}
A_{R,\de}^{(2)}(l,m;z)=A_{D,\de}(h;z)+c_{\de},\ \ \ c_{\de}\in\C,
\ee
with $h$ given by~\eqref{hlm}. 

The equality~\eqref{ARAD} just deduced clearly extends to any pair of modular parameters $a_+,a_-$ in the right half plane satisfying~\eqref{apmz}. Moreover, from the proof of Lemma~3.3 we see that up to addition of a constant any van Diejen A$\De$O can be written as a difference of two twofold products $A_{R,\de}^{(2)}(l,m;z)$. We continue by using~\eqref{Kidj2} to obtain kernel identities for the van Diejen A$\De$Os. To this end, we introduce the vector
\be
\zeta\equiv (1,\ldots,1)=\sum_{n=1}^8 e_n,
\ee
where $e_1,\ldots,e_8$ are the standard basis vectors of $\C^8$. Now we define two maps
\be
\chi_1\colon \C^8\to\C^8,\ \ \ h\mapsto \hat{h}\equiv -h+ \langle \zeta,l\rangle \zeta/4,
\ee
\be
\chi_2\colon \C^8\to\C^8,\ \ \ h\mapsto \tilde{h}\equiv \hat{h}+\zeta/2,
\ee
where $\langle \cdot,\cdot\rangle$ denotes the standard bilinear form on~$\C^8$. It is easily seen that $\chi_1$ is a reflection, whereas
\be
\chi_2(\chi_2(h))=h+\zeta.
\ee
We are now prepared for the following result.

\begin{theorem}
Assume $h\in\C^8$ satisfies~\eqref{hmu} and $a_+,a_-$ are numbers in the right half plane satisfying~\eqref{apmz}.
 Then we have
\be\label{KidjD}
 A_{D,\de}(h;z)K_j(z,y)=A_{D,\de}(\chi_j(h);y)K_j(z,y),\ \ \ \de=+,-,\ \ \ j=1,2.
\ee
\end{theorem} 
\begin{proof}

With the assumption~\eqref{apmz} in effect, the A$\De$Os and kernel functions are continuous in $a_{+}$ and $a_{-}$. 
Thus it suffices to prove the theorem under the stronger restriction that
 $a_{\pm}$ satisfy~\eqref{apmgen}. We first show that there exist four constants~$C_{\pm,j}\in\C, j=1,2$, such that
\be\label{KidjDC}
\big(A_{D,\de}(h;z)-A_{D,\de}(\chi_j(h);y)\big)K_j(z,y)=C_{\de,j}K_j(z,y),
\ee
and then prove that the constants vanish.
Clearly, we need only show~\eqref{KidjDC} for the special $h$ of the form~\eqref{hlm}, for which~\eqref{ARAD} holds true. (Indeed, the case of general $h$ then follows by linearity.) For the special $h$ we may invoke~\eqref{Kidj2}, deducing that the A$\De$O on its right-hand side is of the form $A_{D,\de}(h^{(j)};y)$ plus a constant, with
\be
h^{(j)}_n=\phi_j(m)_n-ia_{-\de}/4,\ \ \ h^{(j)}_{n+4}=\phi_j(l)_n+ia_{-\de}/4,\ \ \ n=1,2,3,4.
\ee
From~\eqref{defkap}--\eqref{phi2} we infer that $h^{(j)}$ is obtained from $\chi_j(h)$ by a permutation in~$S_8$. Since $V_{\de}(h;z)$ is clearly invariant under permutations,  we deduce~\eqref{KidjDC}. 

In order to show that the constants vanish, we combine a few results from Subsection~3.1 of~\cite{Rui09}. The crux is that the constants amount to the shift functions $\sigma_{\pm}(h)$ in~{\it loc.~cit.}, but not for the $z_t$-choice already mentioned in the paragraph below~\eqref{cEt}, but for the choice $\omega_{t,\de}$ we have made in~\eqref{Vbe}. The latter choice also arose in Eq.~(3.53) of~{\it loc.~cit.}, and it follows from the  reasoning leading from this equation to Eq.~(3.58) that the shift constants vanish for this choice.
\end{proof}

 We would like to add that the `$E_8$-identity' encoded in Eq.~(3.55) of~\cite{Rui09}  entails
\be
\sum_{t=0}^3 p_{t,\de}(h)=
\sum_{t=0}^3 p_{t,\de}(\chi_j(h)),\ \ \ j=1,2.
\ee
Hence the additive constants in the two A$\De$Os in~\eqref{KidjD} are equal and can be omitted. (Indeed, $\cE_{t,\de}(\xi;\omega_{t,\de})$ does not depend on $t$, cf.~\eqref{cEcon}.) Doing so, we no longer need to require that $a_{\pm}$ be such that the numbers $R_{\pm}(ia)$ do not vanish.

Possibly, the results just mentioned can also be obtained by a more refined analysis of the constants $c_{\pm}$ in the equality~\eqref{ARAD}. In this connection we point out that for the van Diejen A$\De$Os as defined above (with~\eqref{apmz} in force to prevent divergencies), the vector space properties are easily established. Indeed, the coefficients $V_{\de}(h;-z)$ are not restricted: With \eqref{hmu} in effect, their multiples vary over all of $V_{2,\de}(\mu)$. Also, we claim that the linear span of the A$\De$Os $A_{D,\de}(h;z)$  does not contain the constant functions,  so that it is eight-dimensional, cf.~Prop.~2.1. 

Using our previous results, it is not hard to prove the claim just made. Assuming first~\eqref{apmgen}, the A$\De$Os are of Sklyanin type, so it follows that, given $\lambda_j\in\C$ and $A_{D,\de}(h_j;z)\in\cV_{2,\de}(\mu)$, $j=1,2$, we have
\be\label{ADsum}
\lambda_1A_{D,\de}(h_1;z)+\lambda_2A_{D,\de}(h_2;z)=\lambda_3A_{D,\de}(h_3;z)+s_{\de},
\ee
for some $\lambda_3,s_{\de}\in\C$ and $A_{D,\de}(h_3;z)\in\cV_{2,\de}(\mu)$. Taking residues at the simple poles of the three elliptic functions involved, this yields
\be
\lambda_1p_{t,\de}(h_1)+ \lambda_2p_{t,\de}(h_2) =\lambda_3p_{t,\de}(h_3).
\ee
Recalling~\eqref{Vbe}, we now conclude
\be
\lambda_1V_{b,\de}(h_1)+ \lambda_2V_{b,\de}(h_2) =\lambda_3V_{b,\de}(h_3),
\ee
with $\lambda_3$ and $h_3$ depending on $h_1,h_2$ and $\de$. Therefore, we have $s_{\de}=0$ in~\eqref{ADsum}. Finally, absence of constants for modular parameters satisfying the weaker restriction~\eqref{apmz} follows by continuity.

To conclude this section, we point out a remarkable consequence of the kernel identities in the previous subsection for a subclass of van Diejen type A$\De$Os. Let us choose
\be\label{hsp}
h_5=ia_{+}/2,\ \ h_6=ia_{+}/2+1/2,\ \ h_7=0,\ \ h_8=1/2.
\ee
Then we deduce from~\eqref{AD}--\eqref{cEt} and the duplication formula~\eqref{dup} that we have
\be\label{RD}
A_{D,+}(h;z)=\frac{\prod_{n=1}^4R_{+}(z-h_n-ia_{-}/2)}{R_{+}(2z+ia_{+}/2)}\exp(-ia_{-}\partial_z)+(z\to -z).
\ee
This implies that we can view the Sklyanin generators $A_{R,+}(l;z)$ as a special type of van Diejen operators, with the parameter $a_-$ in~\eqref{RD} replaced by $a_{-}/2$, cf.~\eqref{AR2}--\eqref{fde}. 

A moment's thought now shows that this observation has the consequence that the special A$\De$Os given by~\eqref{hsp}--\eqref{RD} not only admit the kernel identities~\eqref{KidjD} (with $A_{D,+}(\chi_j(h);y)$ not of this type in general), but also the kernel identities~\eqref{Kidj} with $a_{-}$ replaced by $2a_{-}$ throughout (in particular, in the elliptic gamma functions). Moreover, the $y$-dependent A$\De$O in the latter identities is of the same special type. It should be noted, however, that this state of affairs is not modular invariant.

\renewcommand{\thesection}{A}
\setcounter{equation}{0}

\addcontentsline{toc}{section}{Appendix A. Proof of Lemma 3.2}
\section*{Appendix A. Proof of Lemma 3.2}

We begin by reducing the case $t\ne 0$ to the case $t=0$. To this end we use the bijection~\eqref{Phit} to transform $A$ to $\cM_0$. Thus we calculate
\be
(\Phi_tAg)(z)=\sum_{j=-k}^k\tilde{c}_j(z)(\Phi_tg)(z+j\eta),\ \ \ g\in\cM_t,
\ee
where
\be
\tilde{c}_j(z)=\exp(-j\lambda_t/2)c_j(z+\omega_t),\ \ \ t\in \{1,2,3\}.
\ee 
From this it readily follows that it suffices to prove the lemma for the case $t=0$. (Note in particular that the residues $\tilde{r}_j(0,\ell)$ of $\tilde{c}_j(z)$ are equal to $\exp(-j\lambda_t/2)r_j(t,\ell)$.)

Accordingly we take $t=0$ from now on. Assuming the properties (i)--(iii) for the functions $c_j(z)$ occurring in
\be\label{Ag}
(Ag)(z)=\sum_{j=-k}^kc_j(z)g(z+j\eta),\ \ \ \ g\in\cM_0,
\ee
we infer from~\eqref{i} that $Ag$ is an even meromorphic function. Now we use the regularity of $g(z)$ at integer multiples of $\eta$ and property (ii) to infer that the limit
\be
\lim_{z\to -\ell\eta}(z+\ell\eta)(Ag)(z)=\sum_{j=-k}^kr_j(0,\ell)g((-\ell+j)\eta),
\ee
exists. To prove $Ag\in\cM_0$, it remains to be shown that this limit vanishes. By evenness, it suffices to show this for $\ell\ge 0$. Using~\eqref{iiia}, we see that the limit vanishes for $\ell\ge k$, and that we need only sum over $j\ge 2\ell -k$. Thus we are reduced to showing that 
\be\label{sumr}
\sum_{j=2\ell-k}^kr_j(0,\ell)g((-\ell+j)\eta),\ \ \ 0\le \ell <k,
\ee
vanishes. Rewriting this sum as
\be
r_{\ell}(0,\ell)g(0)+\sum_{m=1}^{k-\ell}\big( r_{\ell+m}(0,\ell)g(m\eta)+ r_{\ell-m}(0,\ell)g(-m\eta)\big),
\ee
we can use evenness of $g$ and~\eqref{iiib} in the guise
\be
r_{\ell+m}(0,\ell)=-r_{\ell-m}(0,\ell),\ \ \ m=0,1,\ldots,k-\ell,
\ee
to deduce that it does vanish. This completes the proof that the three conditions imply~\eqref{AMt}.

Conversely, assume $A$ satisfies \eqref{AMt}. Using evenness of $Ag$ and $g$, we see that~\eqref{Ag} entails
\be
\sum_{j=-k}^k(c_j(z)-c_{-j}(-z))g(z+j\eta)=0,\ \ \ \forall g\in\cM_0.
\ee
Now choose $z=z_0\notin \eta\Z$ such that all coefficients $c_j(z)$ are regular at $z=\pm z_0$. Choosing $g$ with $g(z_0)=1$ and $g(z_0\pm j\eta)=0$ for $j=1,\ldots,k$, we obtain $c_0(z_0)=c_0(-z_0)$. Letting next $j_0<0$, we  choose $g$ with $g(z_0+j_0\eta)=1$ and $g(z_0+j\eta)=0$ for $j\ne j_0$. This yields $c_{j_0}(z_0)=c_{-j_0}(-z_0)$. Hence the first property~\eqref{i} readily follows.

To prove properties (ii) and (iii), consider the Laurent expansions of the nonzero coefficients around $z=z(0,\ell)$:
\be\label{cjexp}
c_j(z-\ell\eta)=a_{jl}z^{n_{jl}}+b_{jl}z^{n_{jl}+1}+\ldots,\ \ \ z\to 0,\ \ \ a_{jl}\ne 0,\ \ \  |j|\le k,\ \ \ \ell\in\Z.
\ee
Fixing $\ell>k$ and $j_0$ with $|j_0|\le k$, we can choose $g(z)$ that vanishes to sufficiently high order at $z=-\ell\eta+j\eta$ for any $j\ne j_0$ with $|j|\le k$ so that the terms in the sum~\eqref{Ag} with $j\ne j_0$ have limit 0 (say) as $z\to -\ell\eta$. Since $Ag$ is by assumption regular at $z=-\ell\eta$, the remaining term $c_{j_0}(z)g(z+j_0\eta)$  has a limit too, and when we choose $g(-\ell\eta+j_0\eta)=1$ (say), we see that we must have $n_{j_0,\ell}\ge 0$ in~\eqref{cjexp}. More generally, this argument shows that $c_j(z)$ is regular at $z=z(0,\ell)$ for all $\ell\in\Z$ with $|\ell|>k$.
 
Next we fix $\ell$ with $|\ell|\le k$ and note that the assumption implies that 
\be\label{Agzs}
(Ag)(z-\ell\eta)=c_{\ell}(z-\ell\eta)g(z)+\sum_{m=1}^{k-\ell}c_{\ell+m}(z-\ell\eta)g(z+m\eta)+\sum_{m=1}^{k+\ell}c_{\ell-m}(z-\ell\eta)g(z-m\eta)
\ee
 has a limit as $z\to 0$. For $m\ne 0$, we may choose $g(z)$ such that the functions $g(z\pm m\eta)$ vanish to sufficiently high order as $z\to 0$ so that the sums have a limit. Letting also $g(0)=1$, it follows that~$c_{\ell}(z-\ell\eta)$ has a limit as $z\to 0$, so that $c_{\ell}(z)$ is regular at $z=z(0,\ell)$. 
 
We proceed to consider the special choice $\ell=0$. Fixing $m=j\in\{ 1,\ldots,k\}$,  we may choose $g$ such that the functions $g(z\pm m\eta)$, $m\ne j$, vanish to sufficiently high order as $z\to 0$ so that the terms not involving $c_{\pm j}$ have a limit. Thus it follows that 
\be\label{2cj}
c_j(z-\ell\eta)g(z+j\eta)+c_{-j}(z-\ell\eta)g(z-j\eta)
\ee
has a limit for $z\to 0$. We are still free to choose the numbers $a$ and $b$ in the Taylor expansion
\be
g(z+j\eta)=a+bz+\ldots,\ \ \ g(z-j\eta)=a-bz+\ldots,\ \ \ z\to 0,
\ee
and when we combine this with~\eqref{cjexp}, we see that~\eqref{2cj} satisfies
\be
a(a_{j\ell}z^{n_{j\ell}}+a_{-j,\ell}z^{n_{-j,\ell}}+\ldots)+bz(a_{j\ell}z^{n_{j\ell}}-a_{-j,\ell}z^{n_{-j,\ell}}+\ldots )+\ldots,\ \ \ z\to 0.
\ee
Letting $a\ne 0$, we infer from existence of the $z\to 0$ limit that there are two cases. Either $n_{j\ell}$ and $n_{-j,\ell}$ are both nonnegative (in which case $c_j(z)$ and $c_{- j}(z)$ are regular at $z(0,\ell)$), or we have
\be
n_{j\ell}=n_{-j,\ell}=:n<0,\ \ \ a_{j\ell}+a_{-j,\ell}=0.
\ee
Clearly, we need only analyze the second case further.

There are two subcases to consider, namely $n<-1$ and $n=-1$. In the first subcase we have a leading term
\be
z^{n+1}\big(a(b_{j\ell}+b_{-j,\ell}) +2ba_{j\ell}\big),\ \ \ a_{j\ell}\ne 0.
\ee
Switching to a $g(z)$ for which $a=0$ and $b\ne 0$, we see that we get a pole at $z=0$, so this subcase is ruled out by assumption. For $n=-1$ we deduce that the coefficients $c_{\pm j}(z)$ have simple poles at $z=z(0,\ell)$ with residues
\be
r_{\pm j}(0,\ell)=a_{\pm j,\ell},\ \ r_j(0,\ell)=-r_{-j}(0,\ell).
\ee
Thus the properties (ii) and (iii) follow for $\ell=0$. 

Turning to the choice $\ell \in\{1,\ldots,k\}$, we conclude from~\eqref{Agzs} in a by now familiar way that  for $m>k-\ell$ the factor $c_{\ell-m}(z-\ell\eta)$ cannot be singular for $z\to 0$. Thus $c_j(z)$ has no pole for $j<2\ell-k$. Likewise, for $\ell\in\{-k,\ldots,-1\}$ we infer that $c_j(z)$ has no pole for $j>k+2\ell$. 
Finally, reasoning just as for $\ell =0$, we deduce from~\eqref{Agzs} that for $m=1,\ldots, k-|\ell|$ the functions $c_{\ell\pm m}(z-\ell\eta)$ have at most simple poles for $z\to 0$, with opposite residues. Therefore, we have now shown that the assumption \eqref{AMt} implies the properties (i)--(iii).


\renewcommand{\thesection}{B}
\setcounter{equation}{0}
\setcounter{theorem}{0}

\addcontentsline{toc}{section}{Appendix B. The Sklyanin relations revisited}

\section*{Appendix B. The Sklyanin relations revisited}

We begin by detailing the algebraic formalism that leads to the A$\De$Os~$A_R$~\eqref{AR} introduced in~\cite{Rai06}. Then we clarify how this is related to the Sklyanin algebra. The proofs of the first two lemmas we need  along the way are relegated to the end of the appendix, so as not to interrupt the flow of the reasoning.

Our starting point is the tensor algebra $T(V)$ over the vector space (cf.~Prop.~2.1)
\be\label{V}
V\equiv \theta(2z)V_1(1)=\{ c\theta(z+\vec{a})\mid c\in\C,\ a\in\C^4, \ \sum_{i=1}^4 a_i=0\}.
\ee
Thus $V$ is the 4-dimensional space of theta functions of order 4 that  yield an elliptic function with periods 1 and $\tau$ upon division by $\theta(2z)$. 

Next, we fix $\eta\in\C$ until further notice, and set
\be\label{fV}
f(\alpha,\beta,\gamma;z)\equiv \theta(z+\alpha, z+\beta,z+\gamma,z-\alpha-\beta-\gamma),\ \ \ \alpha,\beta,\gamma\in\C.
\ee
Introducing $F_{\eta}\in V\otimes V$ by
\be\label{Feta}
F_{\eta}(a,b,c_1,c_2;z_1,z_2)\equiv f(a-u,b-u,c_1+u;z_1)f(a+u,b+u,c_2-u;z_2)- (c_1\leftrightarrow c_2),\ \ \ u\equiv\eta/2,
\ee
we define a subspace~$R_{\eta}$ of $V\otimes V$ as the linear span of the functions $F_{\eta}$ with the parameters $a,b,c_1,c_2$ ranging over~$\C$. One readily verifies that this subspace is elliptic in~$\eta$:
\be\label{Rell}
R_{\eta+1}=R_{\eta+\tau}=R_{\eta}.
\ee

Denoting the two-sided ideal in $T(V)$ generated by the functions in~$R_{\eta}$ by~$[R_{\eta}]$, we now introduce the quotient algebra
\be\label{Q}
Q_{\eta}\equiv T(V)/[R_{\eta}].
\ee
In words, we may view the functions in $V$ as the generators of the algebra~$Q_{\eta}$, and the vanishing of all functions in~$R_{\eta}$ as the quadratic $\eta$-dependent relations among the generators. (The relations can also be found in~\cite{Rai06}, using somewhat different conventions.) Note that $Q_{\eta}$ is elliptic in $\eta$, since the relation space $R_{\eta}$ is elliptic, cf.~\eqref{Rell}.

On first encounter, the quotient algebra $Q_{\eta}$~\eqref{Q} may seem very different from the Sklyanin algebra. In fact, however, it is the same for generic $\eta$. We shall prove that in due course, a crucial element of the proof being the next lemma.

\begin{lemma}
For all $\eta\in\C$, the relation space $R_{\eta}$ is 6-dimensional.
\end{lemma}

We postpone the proof of this lemma to the end of this appendix. On the other hand, it may be illuminating to mention at this point that the reasoning runs along the following lines. First, we show that the dimension does not depend on~$\eta$. Then we focus on the special case $\eta=0$. It is plain from the definitions~\eqref{fV}--\eqref{Feta} that all of the tensors~$F_0$ are antisymmetric. Since~$V$ is 4-dimensional, the antisymmetric subspace in $V\otimes V$ is 6-dimensional. Now it is not obvious, but true that a special choice of the parameters guarantees that we get 6 linearly independent $F_0$'s, which constitutes the last step in our argument.

We proceed to introduce a map $\cA_{\nu}$ from the algebra $T(V)$ to an algebra of analytic difference operators, as follows. First, we define the image of a constant $c\in\C\subset T(V)$ as the $c$-multiple of the identity operator. Then we define $\cA_{\nu}$ on the generators by
\be\label{cAnu}
\cA_{\nu}(f)\equiv \frac{f(z-\nu)}{\theta(2z)}\exp(\eta\partial_z)+(z\to -z),\ \ \ \nu\in\C,\ \ \ f\in V.
\ee
Extending this map in the natural way to tensor products, we obtain in particular for functions~$K(z_1,z_2)$ in the 16-dimensional space~$V\otimes V$ the A$\De$Os
\be\label{AVV}
\cA_{\nu}(K)=\frac{K(z-\nu,z-\nu+\eta)}{\theta(2z,2z+2\eta)}\exp(2\eta\partial_z)+(z\to -z)+c_K(z),
\ee
where $c_K(z)$ is the even elliptic function
\be\label{cK}
c_K(z)\equiv \frac{1}{\theta(2z)}\left( \frac{K(z-\nu,-z-\nu-\eta)}{\theta(-2z-2\eta)}-(z\to -z)\right).
\ee 

Obviously, the A$\De$O on the right-hand side of~\eqref{cAnu} equals the A$\De$O $A(a,\nu)$~\eqref{Apar}, but it is not clear at face value that this map gives rise to a representation of the algebra~$Q_{\eta}$. Indeed, for this to be the case, it is necessary and sufficient that~$\cA_{\nu}$ annihilate the relation subspace $R_{\eta}$ of $V\otimes V$. Now it is easy to verify from~\eqref{fV}--\eqref{Feta} that we have
\be\label{AFeta}
F_{\eta}(a,b,c_1,c_2;\de z-\nu,\de z-\nu+\eta)=0,\ \ \ \de=+,-,
\ee
so that
\be\label{AcF}
\cA_{\nu}(F_{\eta})=c_{F_{\eta}}(z).
\ee
But the validity of the following lemma is not obvious.

\begin{lemma}
We have
\be\label{cFzero}
\cA_{\nu}(F_{\eta})=0,\ \ \ \forall\nu\in\C.
\ee
\end{lemma}

Again, we relegate the proof of this lemma to the end of this appendix. In view of~\eqref{AcF}, it amounts to showing that the elliptic function~$c_{F_{\eta}}(z)$ vanishes identically. At this stage we only point out that it follows from Lemma~3.3 that it has no poles, and therefore must be constant. It is not immediate that this constant equals zero, however.

To continue, we introduce the space of (Casimir) functions
\be\label{cCeta}
\cC_{\eta}\equiv \{ c\theta(z_1-z_2\pm\eta,z_1+z_2\pm \gamma)\mid c,\gamma\in\C\}.
\ee
Clearly, any function in $\cC_{\eta}$ belongs to~$V\otimes V$. It is also not hard to see that~$\cC_{\eta}$ is a 2-dimensional subspace of~$V\otimes V$. (Indeed, given $c_1,c_2,\gamma_1,\gamma_2\in\C$ this amounts to the existence of $c_3,\gamma_3\in\C$ such that
\be
c_1\theta(z\pm\gamma_1)+c_2\theta(z\pm\gamma_2)=c_3\theta(z\pm \gamma_3).
\ee
Dividing this by~$\theta(z\pm 1/4)$ (say), and appealing to ellipticity, this is easily checked.) Furthermore, $\cC_{\eta}$ is elliptic in~$\eta$.

Using~\eqref{AVV}, we now calculate
\be\label{Cas}
\cA_{\nu}(\theta(z_1-z_2\pm\eta,z_1+z_2\pm \gamma))=2\theta(2\nu+\eta\pm\gamma).
\ee
This implies in particular
\be\label{AnuC}
\cA_{\nu}(\cC_{\eta})=\C.
\ee
Moreover, fixing $\gamma\in\C$, the constant on the right-hand side of~\eqref{Cas} is nonzero for generic $\nu$. From this we deduce
\be\label{CReta}
\cC_{\eta}\cap R_{\eta}=\{ 0\}.
\ee
We are now prepared for the following lemma.

\begin{lemma}
Let $\eta\in\C^{*}$ and $F\in V\otimes V$. Then we have
\be\label{dim9}
\dim (\cA_{\nu}(V\otimes V))= 9,\ \ \ \forall \nu\in\C,
\ee
\be\label{FR}
\cA_{\nu}(F)=0,\ \forall \nu\in\C \Leftrightarrow F\in R_{\eta}.
\ee
\end{lemma}
\begin{proof}
Since $\cA_{\nu}$ maps $R_{\eta}$ to 0 by Lemma~B.2 and $R_{\eta}$ is 6-dimensional by Lemma~B.1, it follows from~\eqref{AnuC}--\eqref{CReta} that we have
\be
\dim (\cA_{\nu}(V\otimes V))\le 9,\ \ \ \forall \eta,\nu\in\C.
\ee
For $\eta\ne 0$ this dimension count can be sharpened by noting that
\be
\cA_{\nu}(V\otimes V)=\cV_2(\mu),\ \ \ \eta\in\C^{*},\ \ \ \mu =\exp(8i\pi\nu).
\ee
Indeed, from Lemma~3.3 it follows that $\cV_2(\mu)$ is 9-dimensional for any nonzero~$\eta$. Thus, \eqref{dim9} holds true.

To prove the equivalence~\eqref{FR}, let $F\in R_{\eta}$. Then $F$ is a linear combination of the functions $F_{\eta}$, so the implication follows from Lemma~B.2. Finally, assume $\cA_{\nu}(F)$ vanishes for all $\nu$. By virtue of~\eqref{dim9}, the kernel of $\cA_{\nu}$ restricted to $V\otimes V$ is 7-dimensional. Thus it is spanned by the 6-dimensional subspace $R_{\eta}$ and the function~$\theta(z_1-z_2\pm\eta,z_1+z_2\pm (2\nu+\eta))$, cf.~\eqref{Cas}. The latter function depends on~$\nu$, whereas $F$ is constant by assumption. Thus $F$ must belong to $R_{\eta}$, completing the proof.
\end{proof}

We now turn to the connection of the quotient algebra $Q_{\eta}$ defined by~\eqref{Q} with the Sklyanin algebra~$\cS_{\eta}$, as defined by the relations~\eqref{Sk1}--\eqref{Sk2}. First, both algebras are elliptic in their dependence on~$\eta$, so we may as well focus on the period cell spanned by the numbers 1 and $\tau$. For the special case $\eta=0$, the tensors~\eqref{Feta} are all antisymmetric, and by Lemma~B.1 they span the antisymmetric subspace of $V\otimes V$. Thus $Q_0$ equals the commutative algebra of symmetric tensors:
\be\label{Q0}
Q_0=T(V)_s\sim \C[x_1,x_2,x_3,x_4].
\ee

Since the structure constants~$J_{lm}(\eta)$ in the relations~\eqref{Sk1} all vanish for $\eta=0$ (cf.~\eqref{Jlm}), it follows that $S_0$ is a central element of~$\cS_0$, and then the relations~\eqref{Sk2} entail that $\cS_0$ can be viewed as a central extension of the universal enveloping algebra of $su(2)$. As such, $\cS_0$ is not commutative, hence different from $Q_0$. Furthermore, the relations~\eqref{Sk1} are ill defined for $\eta =1/2,\tau/2$, and $1/2+\tau/2$, because two of the structure constants have poles for these $\eta$-values. By contrast, the definition of $Q_{\eta}$ does not involve any divergencies. 

In view of these preliminary observations, we fix  $\eta\notin\Lambda/2$ until further notice. To connect $\cS_{\eta}$ and $Q_{\eta}$, we begin by observing that in the free associative algebra~$\cF$ generated by $S_0,S_1,S_2$ and $S_3$, the quadratic subspace is spanned by the 16 elements $S_{\alpha}S_{\beta}$, and the relations~\eqref{Sk1}--\eqref{Sk2} can be rewritten as
\be\label{Ml}
\sum_{\alpha,\beta=0}^3 M_{\alpha\beta}^{(l)}S_{\alpha}S_{\beta}=0,\ \ \ l=1,\ldots,6.
\ee
It is easy to check that $M^{(1)},\ldots,M^{(6)}$ are 6 linearly independent vectors in $\C^{16}$.

With an eye on~\eqref{fRth}--\eqref{D3}, we now define a map $I$ from the generators $S_t$ to a base for the 4-dimensional space $V$~\eqref{V}:
\be\label{I0}
S_0\mapsto iq^{1/4}G^{-3}\theta(\eta)f((0,1,\tau)/2;z)=\theta_1(\eta)\theta_1(2z),
\ee
\be\label{I1}
S_1\mapsto -iq^{1/4}G^{-3}\theta(\eta+1/2)f((1,-1,1+2\tau)/4;z)=\theta_2(\eta)\theta_2(2z),
\ee
\bea\label{I2}
S_2  &  \mapsto &  iq^{1/4}G^{-3}\exp(i\pi\eta)\theta(\eta+1/2+\tau/2)f((1+\tau,1-\tau,-1+\tau)/4;z)
\nonumber \\
  &  &  =i\theta_3(\eta)\theta_3(2z),
\eea
\bea\label{I3}
S_3  &  \mapsto  &   iq^{1/4}G^{-3}\exp(i\pi\eta)\theta(\eta+\tau/2)f((\tau,-\tau,2+\tau)/4;z)
\nonumber \\
 & & =\theta_4(\eta)\theta_4(2z).
\eea
(Here, we used the duplication formula~\eqref{dup} and~\eqref{th3}--\eqref{th24} to simplify the right-hand sides.) This map naturally extends to an isomorphism from the algebra~$\cF$ to the tensor algebra $T(V)$. In this picture, the 6 relations~\eqref{Sk1}--\eqref{Sk2} are encoded in the vanishing of the functions
\be
m_{\eta}^{(l)}(z_1,z_2)\equiv \sum_{\alpha,\beta=0}^3 M_{\alpha\beta}^{(l)}I(S_{\alpha})(z_1)I(S_{\beta})(z_2),\ \ \ l=1,\ldots,6.
\ee
We are now in the position to state and prove the main result of this appendix.

\begin{theorem}
Let $\eta\notin\Lambda/2$. Then the algebras $\cS_{\eta}$ and $Q_{\eta}$ are isomorphic.
\end{theorem}
\begin{proof}
Comparing \eqref{I0}--\eqref{I3} to \eqref{D0}--\eqref{D3}, we deduce
\be
\cA_{\nu}(I(S_t))=D_t,\ \ \ t=0,1,2,3.
\ee
Now the A$\De$Os $D_t$ satisfy the Sklyanin relations~\eqref{Sk1}--\eqref{Sk2}, so we obtain
\be
\cA_{\nu}(m_{\eta}^{(l)})=0,\ \ \ l=1,\ldots,6,\ \ \ \forall \nu\in\C.
\ee
By virtue of Lemma~B.3, this implies that the 6 functions~$m_{\eta}^{(l)}$ belong to $R_{\eta}$.
Since they are linearly independent and $R_{\eta}$ is 6-dimensional by Lemma~B.1, they span $R_{\eta}$. Thus, the map~$I$ gives rise to an isomorphism of $\cS_{\eta}$ and $Q_{\eta}$. 
\end{proof}

It is clear from the definition of the 6-dimensional relation space $R_{\eta}$ that it is continuous in $\eta$. If we reinterpret 
 the Sklyanin relations as being encoded in the 6-dimensional subspace of $\C^{16}$ spanned by the vectors~$M^{(1)},\ldots,M^{(6)}$ (cf.~\eqref{Ml}), the theorem just proved suggests how $\cS_{\eta}$ should be defined for the 3 excluded
 $\eta$-values: We should multiply the two relations~\eqref{Sk1} involving the two divergent structure constants $J_{lm}$ by the denominator that vanishes, so as to get two finite limit vectors. Then we see that two anticommutators and one commutator vanish. Specifically, we get
\be\label{e1}
\eta=1/2: [S_1,S_2]_+=[S_3,S_1]_+=[S_0,S_1]_-=0,
\ee
\be\label{e2}
\eta=1/2+\tau/2: [S_1,S_2]_+=[S_2,S_3]_+=[S_0,S_2]_-=0,
\ee
\be\label{e3}
\eta=\tau/2: [S_2,S_3]_+=[S_3,S_1]_+=[S_0,S_3]_-=0.
\ee
Combining these relations with~\eqref{Sk2}, we arrive again at 6 linearly independent vectors in~$\C^{16}$. 
Put differently, the 6-dimensional $\eta$-dependent subspace of $\C^{16}$ has a 6-dimensional limit as $\eta$ goes to the 3 excluded values.  Note that it follows from~\eqref{e1}--\eqref{e3} that the resulting 3 algebras are isomorphic via cyclic permutation of the generators $S_1,S_2$ and $S_3$. 

For $\eta=0$ the subspace viewpoint on the relations~\eqref{Sk1}--\eqref{Sk2} yields the same result as before: Since the 3 vectors corresponding to~\eqref{Sk2} do not depend on $\eta$ and the 3 structure constants in~\eqref{Sk1} vanish for $\eta=0$, we obtain a noncommutative algebra~$\cS_0$ differing from~$Q_0$, cf.~\eqref{Q0}.

For the remaining `bad' values of $\eta$ (given by \eqref{e1}--\eqref{e3}), the
functions~\eqref{I0}--\eqref{I3} are no longer a base, but we can still
derive nontrivial information about $Q_{\eta}$ when we reformulate the
relations in terms of the base
\be\label{bt}
b_t\equiv\theta_{t+1}(2z), \ \ \ t=0,1,2,3.
\ee
Choosing first $\eta=1/2$, the relations~\eqref{Sk2} yield the limits
\be\label{b2}
[b_0,b_2]_+=[b_0,b_3]_+=[b_2,b_3]_-=0,
\ee
whereas multiplication by suitable factors in~\eqref{Sk1} yields limit relations
\be\label{b1}
[b_1,b_2]_+=[b_3,b_1]_+=[b_0,b_1]_-=0.
\ee
The 6 tensors in $V\otimes V$ occurring here are clearly linearly
independent, so they yield a base for $R_{1/2}$. In particular, although we
see that $Q_{1/2}$ is not commutative, it is nearly so: any two elements of
even degree commute.  It is straightforward to check that this
commutativity fails for $\cS_{1/2}$, so again $Q_{1/2}$ and $\cS_{1/2}$ are
not isomorphic.

Repeating this reasoning, we deduce that the algebras $Q_{1/2+\tau/2}$ and $Q_{\tau/2}$ admit relations obtained by cyclic permutation of $b_1,b_2$ and $b_3$ in~\eqref{b2}--\eqref{b1}. Thus the 3 algebras at issue are isomorphic, a conclusion that seems hard to obtain via the generating tensors~$F_{\eta}$~\eqref{Feta}. 

We point out that we can also recover~\eqref{Q0} when we first express the Sklyanin relations in terms of $b_0,b_1,b_2$ and $b_3$ by using the map $I$ given by~\eqref{I0}--\eqref{I3}. Indeed, from~\eqref{Sk1} we derive~$[b_k,b_l]_-=0$ and upon divison of (the image of)~\eqref{Sk2} by~$\theta(\eta)$ we obtain the limits~$[b_0,b_k]_-=0$. Thus $R_0$ consists of the antisymmetric tensors and~\eqref{Q0} follows again.

The map~$I$ given by~\eqref{I0}--\eqref{I3} transforms the Casimir elements (cf.~\cite{Skl83})
\be
K_0\equiv \sum_{t=0}^3 S_t^2,\ \ \ K_2\equiv \sum_{k=1}^3
\frac{\theta_{k+1}(2\eta)\theta_{k+1}(0)}{\theta_{k+1}^2(\eta)}S_k^2,
\ee
into the functions
\be
K_0(z_1,z_2)=2\theta(z_1-z_2\pm\eta,z_1+z_2,z_1+z_2),
\ee
\be
 K_2(z_1,z_2)=2\theta(z_1-z_2\pm\eta,z_1+z_2+2\eta,z_1+z_2),
\ee
corresponding to the choices $\gamma= 0$ and $\gamma=\eta$ in~\eqref{cCeta}. (This can be shown by using Jacobi's 5-term identity for the theta function, cf.~also~\cite{Spi09} and~\cite{Ros04}.) From~\eqref{Cas} we then conclude
\be
\cA_{\nu}(K_0)=4\theta(2\nu+\eta)^2,\ \ \cA_{\nu}(K_2)=4\theta(2\nu,2\nu+2\eta),
\ee
in agreement with~\cite{Skl83}, cf.~also~\cite{Spi09}.

It transpires from the above that for $\eta\notin\Lambda/2$ the
$Q(\eta)$-representation furnished by $\cA_{\nu}$ amounts to Sklyanin's
representation $\cD(\mu)$ of~$\cS_{\eta}$ following
from~\eqref{D0}--\eqref{D3}. For $\eta\equiv 0$ it is easily checked that
the A$\De$Os~$\cA_{\nu}(f)$~\eqref{cAnu} commute. This is in agreement with
$Q_{\eta}$ being isomorphic to the polynomial algebra~$\C[x_1,x_2,x_3,x_4]$
for $\eta\in\Lambda$ (recall~\eqref{Q0} and ellipticity of~$Q_{\eta}$
in~$\eta$), whereas the Sklyanin relations~\eqref{Sk1}--\eqref{Sk2}
imply~$\cS_{\eta}$ is noncommutative for $\eta\equiv 0$. 

In fact, it is readily seen that $\cS_0$-representations arise by taking suitable limits of the A$\De$O representations. Specifically, introducing renormalized Sklyanin generators 
\be
D_t^r\equiv D_t/\theta(\eta),\ \ \ t=0,1,2,3,\ \ \ \nu\ne 0,\ \ \ \eta/\nu=c\in\C^{*},
\ee
 we have $D_0^r\to 2$ as $\eta,\nu \to 0$ with $c$ fixed, whereas $D_1^r,D_2^r$ and $D_3^r$  become differential
operators, and the four limits satisfy~\eqref{Sk1}--\eqref{Sk2} with $\eta=0$.

For $\eta\equiv 1/2$ it is not hard to check directly that the A$\De$Os
$\cA_{\nu}(b_t)$ satisfy the relations~\eqref{b2}--\eqref{b1}. Indeed, this follows by using
\be
\theta_j(z+1)=-\theta_j(z),\ \ j=1,2,\ \ \theta_j(z+1)=\theta_j(z),\ \ j=3,4.
\ee
In view of~\eqref{e1} and~\eqref{b1} having the same structure, one may ask whether~$\cS_{1/2}$ as defined via the generating relations~\eqref{e1} and~\eqref{Sk2} is represented by $\cA_{\nu}(b_t)$ as well. This is not the case, however. Indeed, taking e.~g.~$k=2,l=3$ and $m=1$, the commutator of $\cA_{\nu}(b_2)$ and $\cA_{\nu}(b_3)$ vanishes, but the anticommutator of $\cA_{\nu}(b_0)$ and $\cA_{\nu}(b_1)$ is not zero.

Likewise, the algebras $\cS_{\eta}$ with $\eta$ congruent to $\tau/2+1/2$ and $\tau/2$ (as defined by~\eqref{Sk2} and the pertinent limit of~\eqref{Sk1}, cf.~\eqref{e2}--\eqref{e3}) are not represented via the A$\De$Os~$\cA_{\nu}(b_t)$.

Taking again $\eta\notin\Lambda/2$, there is still a difference in the $\nu$-dependence of the representations $\cA_{\nu}$ and $\cD_{\mu}$, even though $Q_{\eta}$ and $\cS_{\eta}$ are isomorphic. Indeed, the coefficient of $\exp(\eta\partial_z)$ in the A$\De$O~$\cA_{\nu}(f)$~\eqref{cAnu} varies over~$V_1(\exp(8i\pi\nu))$ (cf.~Section~2), a space that manifestly has period $1/4$ in $\nu$. By contrast, when we add multiples of $1/4$ to the parameter $\nu$ in the representants~$D_t$~\eqref{D0}--\eqref{D3} of the generators~$S_t$, we do not obtain the same A$\De$Os.  

Put differently, the reliance on a 6-dimensional subspace of $V\otimes V$ to define the algebra $Q_{\eta}$, as compared to using 6 quadratic relations between 4 generators to define $\cS_{\eta}$, gives rise to a slightly different perspective on the A$\De$O representations, a distinction we have glossed over in the main text.

This also applies to the kernel identities in Subsection~4.2. If we specialize the two maps $\Phi_1$ and $\Phi_2$ given by~\eqref{Phi_j} to the A$\De$Os $D_{t,+}$, then we arrive at the anti-automorphisms 
\be\label{aa1}
S_t\mapsto S_t,\ \ \ t=0,1,2,\ \ \ S_3\mapsto -S_3,
\ee
\be\label{aa2}
S_t\mapsto S_t,\ \ \ t=0,1,3,\ \ \ S_2\mapsto -S_2,
\ee
respectively. (This follows by some tedious calculations we omit.) By contrast, the maps $\phi_1$ and $\phi_2$ given by~\eqref{phi1}--\eqref{phi2} simply yield two involutions on the coefficient space~$V_{1,+}(\mu)$.

In one respect the kernel identities do reflect the nontrivial monodromy in $\mu$ of the A$\De$Os~$D_t$ (as $\mu$ circles the origin): When we increase $\nu$ by $1/4$, then the two kernel functions $K_1$ and $K_2$ are not invariant, but get interchanged, cf.~\eqref{gp}--\eqref{K12}. Note that this state of affairs has no bearing on the finite-dimensional submodules of the space $\cM_e$, since they only arise for a discrete set of $\mu$-values.

We conclude this appendix by presenting the proofs of Lemmas~B.1 and~B.2.

\vspace{5mm}

\noindent
{\bf Proof of Lemma B.1.}\ 
From \eqref{fV}--\eqref{Feta} we deduce
\be
F_{\eta}=\theta(z_1+a-\eta/2, z_1+b-\eta/2,z_2+a+\eta/2,z_2+b+\eta/2)\Theta(a+b,\eta/2),
\ee
where
\be\label{Th}
\Theta(d,u) \equiv\theta(z_1+c_1+u,z_1-d-c_1+u,z_2+c_2-u,z_2-d-c_2-u)-(c_1\leftrightarrow c_2).
\ee
We now use the 3-term identity
\be
\theta(z\pm\alpha,\beta\pm\gamma)+\theta(z\pm\beta,\gamma\pm\alpha)+\theta(z\pm \gamma,\alpha\pm\beta)=0,
\ee
to rewrite~\eqref{Th}. Specifically, we set
\be
z=z_1+u-d/2,\ \ \alpha=c_2+d/2,\ \ \beta=c_1+d/2,\ \ \gamma=z_2-u-d/2,
\ee
so that the right-hand side of~\eqref{Th} becomes
\be
\theta(z_1+z_2-d,z_1-z_2+2u,c_2+c_1+d,c_1-c_2).
\ee
For generic $a,b,c_1,c_2\in\C$, therefore, $F_{\eta}$ equals a nonzero constant times the product function
\be
P_{\eta}\equiv\theta(z_1+a-\eta/2, z_1+b-\eta/2,z_2+a+\eta/2,z_2+b+\eta/2,z_1+z_2-a-b,z_1-z_2+\eta).
\ee

Next, we note that $R_{\eta}$ can be viewed as the span of the functions $P_{\eta}$ as the parameters $a,b,c_1,c_2$ range over $\C$. Furthermore, since the factor $\theta(z_1-z_2+\eta)$ does not depend on the latter parameters, the dimension of $R_{\eta}$ does not change when we omit this factor of $P_{\eta}$. Subsequently, we can replace $z_1-\eta/2$ by $y_1$ and $z_2+\eta/2$ by $y_2$ to deduce that $\dim(R_{\eta})$ does not depend on $\eta$.

As a consequence, it suffices to prove 
\be\label{dim6}
\dim(R_0)=6.
\ee
To this end we fix 4 distinct numbers $l_1,l_2,l_3,l_4$ in a period cell with a sum that is not an integer. We claim that we then get a base for $V$ by setting
\be\label{base}
e_1=t(l_2,l_3,l_4),\ e_2=t(l_1,l_3,l_4),\ e_3=t(l_1,l_2,l_4),\ e_4=t(l_1,l_2,l_3),
\ee
where
\be\label{deft}
t(p,q,r)\equiv \theta(z-p,z-q,z-r,z+p+q+r).
\ee

Taking this claim for granted, it follows from~\eqref{fV}--\eqref{Feta} with~$\eta=0$ that the tensors
\be
e_m\otimes e_n -
e_n\otimes e_m,\ \ \ 1\le m<n\le 4,
\ee
are functions of the form $F_0$. (Indeed, two among the numbers $l_j$ are the same.) Since these tensors form a base for the 6-dimensional subspace of antisymmetric tensors in~$V\otimes V$, we now deduce~\eqref{dim6}.

It remains to prove the above claim. We first note that we have
\be
e_m(l_m)=:\lambda_m\ne 0,\ \ \ m=1,2,3,4,
\ee
since the 4 numbers in the period cell are distinct and have a non-integer sum. Moreover,
\be
e_m(l_n)=\de_{mn}\lambda_n,
\ee
as follows from~\eqref{base}--\eqref{deft}. Now let $g\in V$ and consider
\be
h(z):=\sum_{m=1}^4 \frac{g(l_m)}{\lambda_m}e_m(z)\in V.
\ee
Then the difference $d(z)=h(z)-g(z)$ vanishes for all $z_n$ congruent to $l_n$,\ $n=1,2,3,4$. Since $d(z)$ belongs to $V$, its zeros in a period cell must sum to an integer. But the sum of the $l_j$ is not an integer, so we must have $h=g$, and the proof is complete.\ \ \ \ $\Box$
\vspace{5mm}

\noindent
{\bf Proof of Lemma B.2.}\ In view of~\eqref{AcF} it suffices to show that the function~\eqref{cK} with~$K$ replaced by $F_{\eta}$ vanishes identically. From Lemma~3.3 we can deduce that it is constant as a function of $z$, but it seems not easy to choose a special $z$-value for which its vanishing becomes manifest. Therefore, we proceed in a different way. Replacing $\nu -\eta/2$ by $\kappa$, the function in brackets is given by
\bea\label{Fvan}
&  &  \Big[\frac{-1}{\theta(2z+2\eta)}\theta(z+a-\kappa,z+b-\kappa,z+c_1-\kappa +\eta,z-a-b-c_1-\kappa +\eta) 
\nonumber \\
&  &  \times \theta(z-a+\kappa,z-b+\kappa,z-c_2+\kappa+\eta,z+a+b+c_2+\kappa+\eta) -\big( c_1\leftrightarrow c_2\big) \Big] 
\nonumber \\
 &  & -\Big[ z\to -z\Big]. 
\eea
It now suffices to show that this function is zero. 

In order to prove that~\eqref{Fvan} vanishes identically, we view it as the sum of four functions of $\eta$, depending on parameters $z,a,b,c_1,c_2$ and $\kappa$. The crux is that in view of the duplication formula~\eqref{dup} each of the functions is elliptic in $\eta$, with poles occurring only at $\eta\equiv\pm z+\omega_t$, $t=0,1,2,3$. These poles are simple and it is straightforward to check that the residues cancel. Thus~\eqref{Fvan} is constant in~$\eta$. Letting $z\notin \Lambda/2$ and choosing $\eta=0$, we readily see that the constant vanishes.\ \ \ $\Box$

\vspace{10mm}


\noindent
{\Large\bf Acknowledgments}

\vspace{1cm}

\noindent
This collaboration was begun while the authors were visiting the Liu Bie Ju
Centre at City University of Hong Kong. We would like to thank the Centre
and M.~Ismail for the invitation, hospitality and financial support.  The
first author was supported in part by a grant (DMS-1001645) from the
National Science Foundation. Finally, thanks are due to the referee for a careful report, which helped to improve the exposition.

\vspace{5mm}


\bibliographystyle{amsalpha}

\end{document}